\documentclass[11pt]{article}

\usepackage[hang]{footmisc}

\usepackage[hidelinks]{hyperref}
\usepackage{mathtools} %

\usepackage{ mathdots }

\usepackage{multicol} \usepackage{picinpar}

\usepackage{multirow}

\usepackage[normalem]{ulem}

\usepackage{bold-extra}

\usepackage{ifthen} \usepackage{latexsym} \usepackage{amsfonts}
\usepackage{amsmath} \usepackage{amssymb}

 \usepackage{amsthm}

\usepackage{pifont}

\def\desclabel#1{\bf #1\hfil} \def\desc{\list{}{%
    \setlength{\leftmargin}{0pt} \labelwidth= \leftmargin \advance
    \labelwidth by -\labelsep \let \makelabel=\desclabel}}

\def\descHACKlabel#1{\bf #1\hfil} \def\descHACK{\list{}{%
    \setlength{\leftmargin}{0pt} \labelwidth= \leftmargin \advance
    \labelwidth by -\labelsep \let \makelabel=\descHACKlabel}}

\newcounter{extremeleftlistcounter} %
{\begin{list}{\arabic{extremeleftlistcounter}~~~}{\usecounter{extremeleftlistcounter}%
      \setlength{\labelsep}{0pt}\setlength{\leftmargin}{0pt}%
      \setlength{\labelwidth}{0pt}\setlength{\listparindent}{0pt}}}%
  {\end{list}}

\newcounter{leftlistcounter} %
{\begin{list}{\arabic{leftlistcounter}~~~}{\usecounter{leftlistcounter}%
      \setlength{\labelsep}{0pt}\setlength{\leftmargin}{15pt}%
      \setlength{\labelwidth}{15pt}\setlength{\listparindent}{0pt}}}%
  {\end{list}}

\newcommand{\ppznote}[1]{} \newcommand{\anznote}[1]{}
\newcommand{\lhznote}[1]{} \newcommand{\lhzselfnote}[1]{}
\newcommand{\mjznote}[1]{}
\newcommand{\zfootnote}[1]{}

    \newlength{\filength}
 \newsavebox{\gcbox}
\sbox{\gcbox}{\framebox[\filength]{\rule{0ex}{2ex}}}

\newtheorem{theorem}{Theorem}[section]
\newtheorem{corollary}[theorem]{Corollary}
\newcommand\qedblob{\ding{113}} \def\literalqed{{\
    \nolinebreak\hfill\mbox{\qedblob\quad}}}

\newtheorem{lemma}[theorem]{Lemma}
\newtheorem{fact}[theorem]{Fact}
\newtheorem{proposition}[theorem]{Proposition}
\newtheorem{definition}[theorem]{Definition}
\newcount\hour \newcount\minutes \hour=\time \divide\hour by 60
\minutes=\hour \multiply\minutes by -60 \advance\minutes by \time
\def\mmmddyyyy{\ifcase\month\or Jan.\or Feb.\or Mar.\or Apr.\or
  May.\or June\or Jul.\or Aug.\or Sep.\or Oct.\or Nov.\or
  Dec.\fi\space\number\day, \number\year} \def\hhmm{\ifnum\hour<10
  0\fi\number\hour :%
  \ifnum\minutes<10 0\fi\number\minutes}

\newcommand{\naturalnumber}{\ensuremath{{ \mathbb{N} }}}
\newcommand{\realnumber}{\ensuremath{{ \mathbb{R} }}}
\newcommand{\naturalnumberpositive}{\ensuremath{{ \mathbb{N}^+ }}}
\newcommand{\realnumberpositive}{\ensuremath{{ \mathbb{R}^+ }}}
\newcommand{\realnumberatleastone}{\ensuremath{{ \mathbb{R}^{\geq 1}}}}

\newcommand{\naturalfromtwo}{\ensuremath{{ \mathbb{N}^{\geq 2} }}}

 \newcommand{\sharpp}{{\rm \#P}}
\newcommand{\acc}{\ensuremath{{\text{\rm \#acc}}}}

\newcommand{\parityp}{{\rm \oplus P}} \newcommand{\up}{{\rm UP}}

\newcommand{\upleq}[1]{\ensuremath{{\rm UP}_{{\leq}{#1}}}}
 
\newcommand{\fewp}{{\rm FewP}} \newcommand{\coup}{{\rm coUP}}

\newcommand{\p}{{\rm P}} 
 
\newcommand{\primes}{{\rm PRIMES}} \newcommand{\mersenneprimes}{{\rm
    MersennePRIMES}} \newcommand{\composites}{{\rm COMPOSITES}}
 
\newcommand{\np}{{\rm NP}}

\newcommand{\conp}{{\rm coNP}}

\newcommand{\few}{{\ensuremath{\rm Few}}}

\DeclareMathSymbol{\subsetneq}{\mathbin}{AMSb}{"28}
\DeclareMathSymbol{\supsetneq}{\mathbin}{AMSb}{"29}

\DeclareMathOperator{\slog}{slog}

\newcommand{\card}[1]{{ \mathopen\parallel {#1} \mathclose\parallel }}
\newcommand{\ceiling}[1]{{{\lceil {#1} \rceil}}}
\newcommand{\floor}[1]{{{\lfloor {#1} \rfloor}}}

\newcommand{\sigmastar}{\ensuremath{\Sigma^\ast}}

\newcommand{\calf}{\ensuremath{{\cal F}}}

\newcommand{\bigo}{{\protect\cal O}} \newcommand{\bigoh}{{\protect\cal
    O}}
\newcommand{\condition}{\,\mid \:}

\def\land{{\; \wedge \;}} \def\lor{{\; \vee \;}}

 \newcommand{\spp}{\mbox{\rm SPP}}

\newcommand{\rcnb}[1]{{{\rm RC}_{{#1}}}} \newcommand{\rcb}[1]{{{\rm
      RC}_{\{{#1}\}}}}

\title{%
  Gaps, Ambiguity, and Establishing
  Complexity-Class
  Containments via Iterative Constant-Setting\thanks{Work supported in part by NSF grant CCF-2006496.
Work done in part while Mandar Juvekar, Arian Nadjimzadah, and
Patrick A. Phillips were at the
University of Rochester's Department of
Computer Science.  A preliminary version of this paper
appeared in the
proceedings of the
47th International Symposium on
Mathematical Foundations of Computer Science
(MFCS~2022)~\cite{hem-juv-nad-phi:c:ics}.}}
\author{Lane A. Hemaspaandra\\Department of Computer Science\\University of Rochester\\ Rochester, NY 14627, USA
  \and
  Mandar Juvekar\\Department of Computer Science\\Boston University\\Boston, MA 02215, USA
  \and
  Arian
  Nadjimzadah\\Department of Mathematics\\UCLA\\Los Angeles, CA 90095, USA
\and Patrick A.
  Phillips%
  \\Riverside Research\\Arlington, VA 22202, USA}
\date{September 29, 2021; revised February 9, 2024}

\usepackage[dvipsnames,svgnames,x11names]{xcolor}

\begin{document}
\sloppy

\maketitle

\begin{abstract}
  Cai and Hemachandra used iterative constant-setting to prove that
  $\few \subseteq \parityp$ (and thus that
  $\fewp \subseteq \parityp$). In this paper, we note that there is a
  tension between the nondeterministic ambiguity of the class one is
  seeking to capture, and the
  density (or, to be more precise, the needed ``nongappy''-ness) of
  the easy-to-find ``targets'' used in iterative constant-setting. In
  particular, we show that even less restrictive gap-size upper bounds
  regarding the targets allow one to capture
  ambiguity-limited classes.
  Through a flexible, metatheorem-based approach, we do so for a wide
  range of classes
  including the logarithmic-ambiguity version of Valiant's unambiguous
  nondeterminism class $\up$.
  Our work lowers the bar for what advances regarding the existence of
  infinite, P-printable sets of primes would suffice to show that
  restricted counting classes based on the primes have the power to
  accept superconstant-ambiguity analogues of $\up$.
  As an application of our work, we prove that
the Lenstra--Pomerance--Wagstaff Conjecture implies
  that
all
  $(\bigo(1) + \log\log n)$-ambiguity
  NP sets
  are
in the restricted counting class $\rcnb{\primes}$.

\medskip

\noindent
{\bf Keywords}:
Structural complexity theory,
computational complexity theory,
 ambiguity-limited NP,
restricted counting classes,
P-printable sets.
\end{abstract}

\vfill

\section{Introduction}\label{s:intro}
We show that every NP set of low ambiguity
belongs to
broad collections of restricted counting classes.

We now describe the two
types of complexity classes
just mentioned.  For any set $S \subseteq \naturalnumberpositive$, the
restricted counting class
$\rcnb{S}$~\cite{bor-hem-rot:j:powers-of-two} is defined by
$\rcnb{S} = \{L \condition (\exists f \in \sharpp)(\forall
x\in\sigmastar) [(x\not\in L \implies f(x) = 0) \land (x\in L \implies
f(x) \in S)] \}$,
where $\sharpp$ is Valiant's~\cite{val:j:permanent} counting version of $\np$ (see Sec.~\ref{s:defs}).
That is, a set $L$
is in $\rcnb{S}$ exactly if
there is a nondeterministic polynomial-time Turing machine (NPTM) that
on each string not in $L$ has zero accepting paths and on each string
in $L$ has a number of accepting paths that belongs to the set $S$.
For example, though this is an extreme case,
$\np = \rcnb{\naturalnumberpositive}$.

In the 1970s, Valiant started the study of ambiguity-limited versions
of NP by introducing the class $\up$~\cite{val:j:checking},
unambiguous polynomial time, which in the above notation is simply
$\rcb{1}$.
(The ambiguity (limit) of an NPTM refers to an upper
bound on how
many
\emph{accepting} paths it has as a function
of the input's length.  An NP language falls within
a given level of ambiguity if it is accepted by
some NPTM that happens to satisfy that ambiguity limit.)
More generally, for each function $f: \naturalnumber \to \naturalnumberpositive$ or $f: \naturalnumber \to \realnumberatleastone$, $\upleq{f(n)}$ denotes
the class of languages $L$ for which there is an NPTM $N$ such that,
for each $x$, if $x\not\in L$ then $N$ on input $x$ has no accepting
paths, and if $x\in L$ then $1 \leq \acc_N(x) \leq \lfloor f(|x|) \rfloor$ (where
$\acc_N(x)$ denotes the number of accepting computation paths of $N$
on input $x$).
(Since, for all $N$ and $x$, $\acc_N(x) \in \naturalnumber$, the class $\upleq{f(n)}$
just defined would be unchanged if
$\lfloor f(|x|) \rfloor$ were
replaced by
$f(|x|)$.)
Ambiguity-limited nondeterministic classes whose ambiguity limits
range from completely unambiguous ($\upleq{1}$, i.e., UP) to
polynomial ambiguity (Allender and Rubinstein's class
FewP~\cite{all-rub:j:print}) have been defined and studied.

In this paper, we show that many ambiguity-limited counting
classes---including ones based on
types of
logarithmic ambiguity, loglog
ambiguity, and logloglog ambiguity---are
contained in various collections of restricted counting classes.
We
do so primarily through two general theorems (Theorems~\ref{t:meta1}
and~\ref{t:meta2}) that
help make clear how, as the size of the ``holes'' allowed in
the sets underpinning the restricted counting classes becomes
smaller, %
one can handle more ambiguity.
Building on and generalizing earlier
framings~\cite{bor-hem-rot:j:powers-of-two}, we will quantify a set's lack of large holes as its ``nongappiness.''
Our basic
notion capturing this
(see Definition~\ref{d:nongappy-main}) is that a nonempty set is $F$-nongappy if for each element $m$ in the set there exists an $m'>m$ such that $m'$ also is in the set and satisfies  $|m'| \leq F(|m|)$.
Table~\ref{t:results-summary} summarizes our results about the
containment of ambiguity-limited counting classes in
restricted counting classes.
\begin{table}[!tbp]\small %
  \centering

  \def\arraystretch{1.35}

  \addtolength{\tabcolsep}{-3.5pt}

    \begin{tabular}{l l l}
      \multicolumn{2}{c}{If $T\subseteq \naturalnumberpositive$ X, then Y}&\\
      X  &  Y & Reference \\\hline

      has an $(n+\bigo(1))$-nongappy, P-printable subset&
                                                          $\fewp \subseteq \rcnb{T}$  &\cite{bor-hem-rot:j:powers-of-two}\\

      has an $\bigo(n)$-nongappy, P-printable subset&
                                                      $\upleq{\bigo(\log n)}
                                                    \subseteq \rcnb{T}$
         &Thm.~\ref{t:kn}\\

      has
      an $\bigo(n\log n)$-nongappy, P-printable subset&
                                                        $\upleq{\bigo(\sqrt{\log n})} \subseteq \rcnb{T}$\\[3pt]

      \parbox[c]{2.5in}{for any $c \in \naturalnumberpositive$ has an $n^{2^{c/2}}$-nongappy,
      \mbox{P-printable} subset}&
                                  $\upleq{\bigo(1) + \frac{\log \log n}{c}} \subseteq \rcnb{T}$
         &Thm.~\ref{t:n^k}\\[7pt]

      \parbox[c]{2.5in}{for any $k \in \naturalnumberpositive$ has an $n^{(\log n)^k}$-nongappy,
      \mbox{P-printable} subset}&
                                                                      $\upleq{\bigo(1) + \frac{1}{\ceiling{\log(k+1)+1}}\log \log \log n} \subseteq \rcnb{T}$&   Thm.~\ref{t:tradeoff}\\[6pt]

      has a $2^n$-nongappy, P-printable subset $S$&
                                                    \parbox[c]{2.0in}{$\upleq{\max(1,
      \floor{\frac{ \log^*(n) - \log^*(\log^*(n) + 1) - 1}{\lambda}})}$\\$\subseteq \rcnb{T}$,
      where\\$\lambda = 4 + \min_{s \in S, |s| \geq 2}(|s|)$}&  Thm.~\ref{t:tradeoff}\\

      is infinite & $\upleq{\bigo(1)} \subseteq \rcnb{T}$ & Cor.~\ref{c:infinite-big-o}
    \end{tabular}
    \caption{\label{t:results-summary}Summary of containment results.  (Theorem~\ref{t:tradeoff}
also
gives a slightly stronger form of the
$2^n$-nongappiness result than the version stated here.)}
\end{table}

Only for polynomial ambiguity was a result of this sort previously
known.  In particular, Beigel, Gill, and
Hertrampf~\cite{bei-gil-her:cOUTbybei-gil-jour-and-her-jour:mod},
strengthening Cai and Hemachandra's result
$\fewp \subseteq \parityp$~\cite{cai-hem:j:parity}, proved that
$\fewp \subseteq \rcb{1,3,5,\dots}$, and Borchert, Hemaspaandra, and
Rothe~\cite{bor-hem-rot:j:powers-of-two}
noted that $\fewp \subseteq \rcnb{T}$ for each nonempty set
$T\subseteq \naturalnumberpositive$ that has an easily presented
(formally, P-printable~\cite{har-yes:j:computation}, whose definition
will be given in Section~\ref{s:defs}) subset $V$ that is
$(n+\bigo(1))$-nongappy (i.e., for some $k$ the set $V$ never has
more than $k$ adjacent, empty lengths;  that is,
for each collection of $k+1$ adjacent lengths, $V$ will always contain
at least one string whose length is one of those $k+1$ lengths).

Our proof approach in the present paper
connects somewhat
interestingly to the history just mentioned.  We will describe in
Section~\ref{sec:gaps} the approach that we will call \emph{the
  iterative constant-setting technique}.  However, briefly put, that
refers to a process of sequentially setting a series of
constants---first $c_0$, then $c_1$, then $c_2$,~\dots, and then
$c_m$---in such a way that, for each $0 \leq j \leq m$, the summation
$\sum_{0 \leq \ell \leq j} c_\ell {j \choose \ell}$ falls in a certain
``yes'' or ``no'' target set, as required by the needs of the setting.
For $\rcnb{S}$ classes, the ``no'' target set will be $\{0\}$ and the
``yes'' target set will be $S$.  In this paper, we will typically
put sets
into restricted counting classes
by building Turing machines that guess
(for each $0 \leq \ell \leq j$)
cardinality-$\ell$
sets
of accepting paths of another NPTM and then amplify each such
successful accepting-path-set guess
by---via splitting/cloning of the path---creating
from it $c_\ell$ accepting paths.

A technically novel aspect of the proofs of the two main theorems
(Theorems~\ref{t:meta1} and~\ref{t:meta2}, each in effect
a metatheorem) is that
those proofs
each provide, in a unified way for a broad
class of functions, an analysis of value-growth
in the context of
iterated functions.

Cai and Hemachandra's~\cite{cai-hem:j:parity} result
$\fewp \subseteq \parityp$ was proven (as was an even more general
result about a class known as ``Few'') by the iterative
constant-setting technique.  Beigel, Gill, and
Hertrampf~\cite{bei-gil-her:cOUTbybei-gil-jour-and-her-jour:mod},
while generously noting that ``this result can also be obtained by a
close inspection of Cai and Hemachandra's proof,'' proved the far
stronger result $\fewp \subseteq \rcb{1,3,5,\dots}$ simply
and directly rather than by iterative constant-setting.
Borchert, Hemaspaandra, and
Rothe's~\cite{bor-hem-rot:j:powers-of-two} even more general result,
noted above for its proof, resurrected the iterative constant-setting
technique, using it to understand one particular level of ambiguity.
This present paper is, in effect, an immersion into the far richer
world of possibilities that the iterative constant-setting technique
can offer, if one puts in the work to
analyze and bound the growth rates of certain constants central to the
method. In particular, as noted above we use the iterative
constant-setting method to obtain a broad range of results
(see Table~\ref{t:results-summary})
regarding
how ambiguity-limited nondeterminism is not more powerful than
appropriately nongappy restricted counting classes.

Each of our results has immediate consequences regarding the power of
the primes as a restricted-counting acceptance type.  Borchert,
Hemaspaandra, and Rothe's result implies that if the set of primes has
an $(n+\bigo(1))$-nongappy, P-printable subset, then
$\fewp \subseteq \rcnb{\primes}$.  However, it is a long-open research
issue whether there exists \emph{any}
infinite, P-printable subset of the primes, much less an
$(n+\bigo(1))$-nongappy one.  Our results lower the bar on what one
must assume about how nongappy hypothetical infinite, P-printable
subsets of the primes are in order to imply that some
superconstant-ambiguity-limited nondeterministic version of NP is
contained in $\rcnb{\primes}$.  We prove that even infinite,
P-printable sets of primes with merely exponential upper bounds on
the size of
their gaps would yield such a result.
We also prove---by exploring the relationship between density and
nongappiness---that the Lenstra--Pomerance--Wagstaff
Conjecture~\cite{pom:j:primality-testing,wag:j:mersenne-numbers}
(regarding the asymptotic density of the Mersenne primes) implies that
$\upleq{\bigo(1) + \log \log n} \subseteq \rcnb{\primes}$.
The Lenstra--Pomerance--Wagstaff Conjecture is characterized in
Wikipedia~\cite{wik:url:gillies-conjecture} as being ``widely
accepted,'' the fact that it disagrees with a different
conjecture (Gillies' Conjecture~\cite{gil:j:mersenne})
notwithstanding.

\section{Definitions}\label{s:defs}

$\naturalnumber = \{0,1,2,\dots\}$.
$\naturalnumberpositive = \{1,2,\dots\}$.  Each positive natural
number, other than 1, is prime or composite.  A prime number is a
number that has no positive divisors other than 1 and itself.
$\primes = \{i \in \naturalnumber \condition i$ is a
prime$\} = \{2,3,5,7,11,\dots\}$.  A composite number is one that has
at least one positive divisor other than 1 and itself;
$\composites = \allowbreak\{i \in \naturalnumber \condition i$ is a
composite number$\} = \{4,6,8,9,10,12,\dots\}$.
$\realnumber$ is the set of all real numbers,
$\realnumberpositive = \{x \in \realnumber \condition x > 0\}$,
and $\realnumberatleastone = \{x \in \realnumber \condition x \geq 1\}$.

All logarithms in this paper---including those involved in log, loglog, and logloglog, those invoked by the definitions of $\log^{[i]}$ and $\log^*$ in the next paragraph, and also those in the definition of our new $\log^\circledast$ which appears later---are base~2.
Also, each call of the log function in this paper, $\log(\cdot)$, is
implicitly a shorthand for $\log(\max(1,\cdot))$.  We do this so that
formulas such as $\log\log\log(\cdot)$ do not cause domain problems on
small inputs.  (Admittedly, this is also distorting log in the
domain-valid open interval (0,1).  However, that interval never comes
into play in our paper except incidentally when iterated logs drop
something into it, and also in the definitions of $\log^*$ and
$\log^\circledast$, but in those cases we will argue that the max
happens not to change what those evaluate to there.)
For any function $f$, we use $f^{[n]}$ to denote function iteration:
$f^{[0]}(\alpha) = \alpha$ and inductively, for each
$n\in \naturalnumber$, $f^{[n+1]}(\alpha) = f(f^{[n]}(\alpha))$.
For each real number $\alpha \geq 0$, $\log^*(\alpha)$ (``(base 2) log
star of $\alpha$'') is the smallest natural number $k$ such that
$\log^{[k]}(\alpha) \leq 1$.  Although the logarithm of 0 is not
defined, note that $\log^*(0)$ is well-defined, namely it is 0 since
$\log^{[0]}(0) =
0$.\footnote{\label{f:half-open-interval-logstar}Since the definition
  of $\log^*(\cdot)$ allows values on the interval $[0,1)$, one might
  worry that the fact that we have globally redefined $\log(\cdot)$ to
  implicitly be $\log(\max(1,\cdot))$ might be changing what
  $\log^*(\cdot)$
  evaluates to.  However, it is easy to see that,
  with or without the max, what this evaluates to in the range $[0,1)$
  is 0, and so our implicit max is not changing the value of
  $\log^*$.}

As mentioned earlier, for any NPTM $N$ and any string $x$, $\acc_N(x)$
will denote
the number of accepting computation paths of $N$ on input $x$.
$\sharpp$~\cite{val:j:permanent} is the counting version of NP:
$\sharpp = \{f:\sigmastar \rightarrow \naturalnumber \condition
(\exists$ NPTM $N) (\forall x \in \sigmastar)[\acc_N(x) = f(x)]\}$.
$\parityp$ (``Parity P'') is the class of sets $L$ such that there is
a function $f\in\sharpp$ such that, for each string $x$, it holds that
$x \in L \iff f(x) \equiv 1 \pmod
2$~\cite{pap-zac:c:two-remarks,gol-par:j:ep}.

We will use $\bigo$ in its standard sense, namely, if
$f$ and $g$ are functions (from whose domain negative numbers are
typically excluded), then we say $f(n)=\bigo(g(n))$ exactly if there
exist positive integers $c$ and $n_0$ such that
$(\forall n \geq n_0)[f(n)\leq cg(n)]$.  We sometimes will also,
interchangeably,
speak
of or write a $\bigo$ expression as representing a set of
functions
(e.g., writing $f(n) \in
\bigo(g(n))$)~\cite{bra:j:crusade-asymptotics-big-o,bra-bra:b:algorithmics},
which in fact is what the ``big O'' notation truly represents.

  The notions $\rcnb{S}$, $\up$, and ${}\upleq{f(n)}$ are as defined in
Section~\ref{s:intro}.
For each $k\geq 1$, Watanabe~\cite{wat:j:hardness-one-way} implicitly
and Beigel~\cite{bei:c:up1} explicitly studied the constant-ambiguity
classes
$\rcb{1,2,3,\dots,k}$ which, following the notation of Lange and
Rossmanith~\cite{lan-ros:c:up-circuit-and-hierarchy}, we will usually
denote $\upleq{k}$.
We extend
the definition of $\upleq{f(n)}$
to
classes of functions as follows.
For classes $\mathcal F$
of functions mapping $\naturalnumber$ to $\naturalnumberpositive$ or $\naturalnumber$ to $\realnumberatleastone$,
we define
$\upleq{\mathcal F} = \bigcup_{f \in {\mathcal F}} \upleq{f(n)}$.
We mention that
the class $\upleq{\bigo(1)}$
is easily seen to
be equal to
$\bigcup_{k \in \naturalnumberpositive} \upleq{k}$, which is a good thing
since that latter definition of the notion is how
$\upleq{\bigo(1)}$ was defined in the literature more than a quarter of a century
ago~\cite{hem-zim:tOutByJourExceptUPkStuffOnlyIsHere:balanced}.
$\upleq{\bigo(1)}$
can be (informally) described as the
class of all sets acceptable by NPTMs with constant-bounded
ambiguity.
Other related classes will
also
be of interest to us.  For example,
$\upleq{\bigo(\log n)}$
captures the class of all sets acceptable by
NPTMs with logarithmically-bounded ambiguity.
Allender and Rubinstein~\cite{all-rub:j:print} introduced and studied
$\fewp$, the polynomial-ambiguity NP languages,
which can be defined by
$\fewp = \{ L \condition (\exists ~\textrm{polynomial}~f) [ L \in
\upleq{f(n)}]\}$.

The $\upleq{f(n)}$ classes, which will be central to this paper's
study, capture ambiguity-bounded versions of $\np$.  They are also
motivated by the fact that they completely characterize the
  existence of ambiguity-bounded (complexity-theoretic) one-way
  functions.\footnote{A (possibly nontotal) function $g$ is said to
  be a one-way function exactly if (a)~$g$ is polynomial-time
  computable, (b)~$g$ is honest (i.e., there exists a polynomial $q$
  such that, for each $y$ in the range of $g$, there exists a string
 $x$ such that $g(x)=y$ and $|x| \leq q(|y|)$; simply put, each
  string $y$ mapped to by $g$ is mapped to by some string $x$ that is not
  much longer than $y$), and (c)~$g$ is not polynomial-time
  invertible (i.e., there exists no (possibly nontotal)
  polynomial-time function $h$ such that for each $y$ in the range of
  $g$, it holds that $h(y)$ is defined and $g(h(y))$ is defined and
  $g(h(y)) =y$)~\cite{gro-sel:j:complexity-measures}.  For each
  $f: \naturalnumber \rightarrow \naturalnumberpositive$ and each
  (possibly nontotal) function $g: \sigmastar \rightarrow \sigmastar$,
  we say that $g$ is $f(n)$-to-one exactly if, for each
  $y \in \sigmastar$, $\card{ \{x \condition g(x)=y \}} \leq f(|y|)$.
   When $g$ is a one-way function, the function $f$ is sometimes
   referred to as an ambiguity limit on the function
   $g$, and the special case
   of $f(n)=1$ is the case of unambiguous one-way functions.
   (This is a
   different notion of ambiguity than that used for NPTMs, though
   Proposition~\ref{p:generic-link} shows that the notions are closely connected.)%
  }
\begin{proposition}\label{p:generic-link}
  Let $f$ be any function mapping from $\naturalnumber$ to
  $\naturalnumberpositive$.  $\p \neq \upleq{f(n)}$ if and only if
  there exists an $f(n)$-to-one one-way function.
\end{proposition}
We say a function $f$ is nondecreasing if $n \leq n'$ implies $f(n) \leq f(n')$.
Proposition~\ref{p:generic-link} holds even if $f$ is not nondecreasing, and holds even if
$f$ is not a computable function.  To the best of our knowledge,
Proposition~\ref{p:generic-link} has not been stated before for the
generic case of any function
$f: \naturalnumber \rightarrow \naturalnumberpositive$.  However, many
concrete special cases are well known, and the proposition follows
from the same argument as is used for those (see for example
\cite[Proof of Theorem 2.5]{hem-ogi:b:companion} for a tutorial
presentation of that type of argument).  In particular, the
proposition's special cases are known already for UP (due
to~\cite{gro-sel:j:complexity-measures,ko:j:operators}),
$\upleq{k}$ (for
each $k \in \naturalnumberpositive$)
and $\upleq{\bigo(1)}$ %
(in~\cite{hem-zim:tOutByJourExceptUPkStuffOnlyIsHere:balanced,ber:thesis:iso}),
$\fewp$ (in~\cite{all-rub:j:print}), and (since the following is
another name for $\np$) $\upleq{2^{n^{\bigo(1)}}}$ (folklore, see
\cite[Theorem 2.5, Part 1]{hem-ogi:b:companion}).
The proposition
holds not just for single functions $f$, but also for classes that are
collections of functions, e.g., $\upleq{\bigo(\log n)}$.
We
pause from our presentation of definitions to discuss whether there
even are sets that fall in such classes as $\upleq{2}$ or
$\upleq{\bigo(\log n)}$, yet are not also obviously even in
$\up$. In terms of directly defined, highly concrete, natural examples,
to the best of our knowledge, none are yet
known.
But the lack of currently-known
concrete sets does not mean that the study is without value.
Ambiguity is a natural resource, and complexity tries to better understand the relationships between different model and resource restrictions, such as between limited ambiguity and restricted counting classes.

However, we in fact will now give three different types of
indirect constructions that put sets into, for example, such limited
ambiguity classes as $\upleq{\bigo(\log n)}$.  In each of our
three construction types, there is
no obvious argument that the sets constructed belong to $\up$.  Thus
the approaches are providing candidates sets for, for example,
$\upleq{\bigo(\log n)} - \up$.

One type is implicit in the proofs
underpinning Proposition~\ref{p:generic-link} and the results in
the paragraph that follows it.
Using logarithmic ambiguity as our example,
each $\bigo(\log n)$-to-1 honest, polynomial-time computable function $f$ implicitly (from the construction underpinning Proposition~\ref{p:generic-link}) defines a set $L_f$ that is in $\up_{\leq \bigo(\log n)}$.
(Additionally, $L_f$ has the property that if $L_f \in \p$, then $f$ is polynomial-time invertible.)
We see no obvious way of showing that $L_f$ will be in $\up$.
Similar claims hold for the other density bounds.
So low-ambiguity sets are in fact closely tied, via Proposition~\ref{p:generic-link}, to whether
low-injectivity
(complexity-theoretic) one-way functions exist.

The second type of construction of sets in, for example, $\upleq{\bigo(\log n)}$ comes from looking at downward disjunctive reducibility cones, i.e., taking an ``or'' of a collection of queries to a UP set.
Namely, the class $\mathrm{R}^p_{\bigo(\log n)\text{-dtt}}(\up)$ is clearly contained in $\upleq{\bigo(\log n)}$ (and so is $\mathrm{R}^p_{\bigo(\log n)\text{-T}}(\up)$, since that equals $\mathrm{R}^p_{\bigo(\log n)\text{-dtt}}(\up)$).
Briefly, a set $L$ is in $\mathrm{R}^p_{\bigo(\log n)\text{-dtt}}(\up)$ if there is a UP set $A_L$, such that on input $x$ we can in polynomial time compute a list of $\bigo(\log |x|)$ strings such that $x \in L$ exactly if at least one string in our list belongs to $A_L$.
The class $\mathrm{R}^p_{\bigo(\log n)\text{-T}}(\up)$ is defined the same way except instead of nonadaptive queries the machine can ask $\bigo(\log n)$ sequential---i.e., adaptive---oracle queries to $A_L$, but must accept exactly if at least one belongs to $A_L$.
To make this a bit more concrete, note that from
each $\up$ set, $A$, we get the following
simple example of such an $\mathrm{R}^p_{\bigo(\log n)-\text{dtt}}(\up)$ set,
which by the above comment will also belong to $\upleq{\bigo(\log n)}$:
$\textrm{At-Least-One-of-Short-List}_A = \{  (I_1,I_2,I_3,\ldots,I_k) \condition
k \leq \log(\sum_{1 \leq i \leq k} |I_i|) \land
\{I_1, I_2, I_3, \dots I_k\} \cap A \neq \emptyset \}$,
 i.e., the set of all lists of potential instances of $A$ such that at least one member of the list belongs to $A$ and the number of items in the list is quite small relative to the size of the list's encoding.
If $A$ in fact belongs to $\up \cap \coup$ then $\textrm{At-Least-One-of-Short-List}_A$ in fact is itself in $\up \cap \coup$ (since $\p^{\up \cap \coup} = \up\cap\coup$); so for this example to have any possible chance of escaping $\up$,
we need $A$ to be a set in $\up - \coup$.\footnote{Most complexity theorists probably suspect that $\up \neq \coup$ (equivalently, $\up - \coup \neq \emptyset$), though likely with less conviction than they suspect its famous ambiguity-unbounded analogue, $\np \neq \conp$.  Both those results in fact are known to hold with probability one relative to a random oracle (respectively by Hemaspaandra and Zimand~\cite{hem-zim:j:balanced} and by the seminal paper of Bennett and Gill~\cite{ben-gil:j:prob1}), although that is not known to be determinative of whether they hold in the unrelativized world.
}

There exists a third approach to placing sets within bounded-ambiguity classes,
which comes from Theorem~5 of a paper by Allender and Rubinstein~\cite{all-rub:j:print}.
That approach---which Allender and Rubinstein do for the case of $\fewp$ but which, with the natural adjustment of changing the degree of sparseness, would also apply to our classes---however creates, via prefix sets, \emph{sparse} sets in $\fewp$ (or our other bounded-ambiguity classes).
And that is a higher hurdle than merely putting \emph{some} sets interestingly in our classes.
In fact, the one-way functions approach above completely characterizes whether the classes collapse to $\p$, and the Allender--Rubinstein approach completely characterizes whether the sparse sets in the class collapse to the sparse sets in $\p$.
In both cases, the issue of whether the constructed sets are in $\up$ is an open one; we see no obvious argument that the sets will be in $\up$, but that is not a guarantee.

Returning to definitions, a set $L$ is said to be P-printable~\cite{har-yes:j:computation}
exactly if there is a deterministic polynomial-time Turing machine
such that, for each $n \in \naturalnumber$, the machine when given as input the string
$1^n$ prints (in some natural coding, such as printing each of the
strings of $L$ in lexicographical order, inserting the character \#
after each) exactly the
set of all strings in $L$ of length less than or equal to $n$.

Notions of whether a set has large empty expanses between one element
and the next will be central to our work in this paper.  Borchert,
Hemaspaandra, and Rothe~\cite{bor-hem-rot:j:powers-of-two} defined and
used such a notion, in a way that is tightly connected to our work.
We present here the notion they called ``nongappy,'' but here, we will
call it ``nongappy$_{\textrm{value}}$'' to distinguish their
value-centered definition from the length-centered definitions that
will be our norm in this paper.
\begin{definition}[\cite{bor-hem-rot:j:powers-of-two}]
  A set $S \subseteq \naturalnumberpositive$ is said to be
  nongappy$_{\textrm{value}}$ if $S \neq \emptyset$ and
  $(\exists k>0)(\forall m\in S)(\exists m'\in S) [m'>m\land m'/m \leq
  k]$.
\end{definition}
This says that the gaps between one element of the set and the next
greater one are, as to the \emph{values} of the numbers, bounded by a
multiplicative constant.  Note that, if we view the natural numbers as
naturally coded in binary, that is equivalent to saying that the gaps
between one element of the set and the next greater one are, as to the
\emph{lengths} of the two strings, bounded by an additive constant.
That is, a nonempty set $S\subseteq \naturalnumberpositive$ is said to
be nongappy$_{\textrm{value}}$ by this definition if the gaps
in the lengths of elements of $S$
are
bounded by an
additive constant, and thus we have the following result that clearly
holds.
Note that throughout this paper, for strings $x$, we use $|x|$ to refer to the number of characters in the string $x$ but, as one can see in the following proposition, for natural numbers $m$, we use $|m|$ to refer to the length of the binary representation of $m$.
\begin{proposition} \label{p: nongappy-value}
  A set $S \subseteq \naturalnumberpositive$ is
  nongappy$_{\textrm{value}}$ if and only if $S \neq \emptyset$ and
  $(\exists k>0)(\forall m \in S) (\exists m' \in S) [m' > m \land
  |m'| \leq |m| + k]$.
\end{proposition}
In Section~\ref{sec:gaps}
we define other notions of
nongappiness
that allow
larger gaps than the above does.  We will always focus on lengths, and
so we will consistently use the term ``nongappy'' in our definitions
to speak of gaps quantified in terms of the \emph{lengths} of the strings
involved.  We now introduce a new notation for the notion
nongappy$_{\textrm{value}}$, and show that our definition does in fact
refer to the same notion as that of
Borchert, Hemaspaandra, and Rothe.
\begin{definition} \label{def: n+O(1)-nongappy} A set
  $S \subseteq \naturalnumberpositive$ is $(n+\bigo(1))$-nongappy if
  $S \neq \emptyset$ and
  $(\exists f \in \bigo(1))(\forall m \in S)(\exists m'\in S) [m' > m
  \land |m'| \leq |m| + f(|m|)]$.
\end{definition}

The issue of sets having
nongappy (in various strengths of that notion),
P-printable subsets will be very important
to our paper.  Let us give a simple
example that helps illustrate some of these
notions, and does so by giving a
nongappy, P-printable subset of SAT (naturally
encoded).
It is not known whether all
$(n+\bigo(1))$-nongappy
NP-complete sets have $(n+\bigo(1))$-nongappy, P-printable subsets (or even have any infinite, P-printable subset at all).
(We mention in passing that it is not hard to see there are NP-complete sets
that are not
$(n+\bigo(1))$-nongappy. Those sets trivially cannot have $(n+\bigo(1))$-nongappy subsets, which is why in this example, to be fair, we focus only on $(n+\bigo(1))$-nongappy NP-complete sets.)
However, in its natural encoding, SAT clearly does have $(n+\bigo(1))$-nongappy, P-printable subsets, e.g., $\{``v", ``v \lor v", ``v \lor v \lor v", \ldots\}$, where $v$ is some fixed variable name.
We thus have an example where SAT has a simplicity property (namely, having a $(n+\bigo(1))$-nongappy, P-printable subset) that is not currently known to hold for all
$(n+\bigo(1))$-nongappy
NP-complete sets.

While at first glance Definition~\ref{def: n+O(1)-nongappy} might seem to be different from Borchert,
Hemaspaandra, and Rothe's definition, it is easy to see that both
definitions are equivalent.
\begin{proposition} \label{p: (n + O(1))-ng equiv ng-value}
  A set $S$ is $(n+\bigo(1))$-nongappy if and only if it is
  nongappy$_{\textrm{value}}$.
\end{proposition}
\begin{proof}
  Both directions follow immediately from Proposition~\ref{p:
    nongappy-value}.  In particular, if $S$ is
  $(n + \bigo(1))$-nongappy then there is some function $f$ as in
  Definition~\ref{def: n+O(1)-nongappy}.
  Since $f \in \bigo(1)$,
  clearly there exists a constant $k>0$ such that
  $(\forall m \in
  \naturalnumberpositive)[f(|m|) \leq k]$.
Conversely, if $S$ is
nongappy$_{\textrm{value}}$ then there is a constant $k$ as in
  Proposition~\ref{p: nongappy-value}, and since the constant function
  $k$ is of course $O(1)$, we that have that $S$ is
  $(n + \bigo(1))$-nongappy.
\end{proof}
Finally, the end of the fifth paragraph of Section~\ref{s:intro} gave
an on-the-fly, quite simple characterization of the
$(n+\bigo(1))$-nongappy sets as being the class of all sets
$S \subseteq \naturalnumberpositive$ such that, for some $k>0$, $S$
never has more than $k$ adjacent lengths containing no strings
($k=0$ was not
excluded, but w.l.o.g.~we may assume $k>0$, since if it holds for $k=0$
it holds for $k=1$).  For completeness, we
briefly explain why that indeed is a correct characterization of that
notion.  In particular, in light of Proposition~\ref{p: (n + O(1))-ng
  equiv ng-value}, we need only show that, for each set
$S \subseteq \naturalnumberpositive$, the just-mentioned
characterization holds exactly if the right-hand side of
Proposition~\ref{p: nongappy-value} holds.  If the former holds with
$k$ set to $k'$, then the right-hand side of Proposition~\ref{p:
  nongappy-value} clearly holds with $k$ set to $k'+1$.  If the
right-hand side of Proposition~\ref{p: nongappy-value} holds with $k$
(recall, $k>0$ there) set to $k'$, then the characterization from
Section~\ref{s:intro} clearly holds with $k$ set to
$\max(k'-1,\min_{m \in S}(|m|))$.

\section{Related Work}

The most closely related work has already largely been covered in the
preceding sections, but we now briefly mention
that work and its
relationship to this paper.  In particular, the most closely related
papers are the work of Cai and Hemachandra~\cite{cai-hem:j:parity},
Hemaspaandra and Rothe~\cite{hem-rot:j:rice}, and Borchert,
Hemaspaandra, and Rothe~\cite{bor-hem-rot:j:powers-of-two}, which
introduced and studied the iterative constant-setting technique as a
tool for exploring containments of counting classes.  The former two
(and also the
important related work of Borchert and Stephan~\cite{bor-ste:j:rice})
differ from the present paper in that they are not about restricted
counting classes, and unlike the present paper,  Borchert, Hemaspaandra,
and Rothe's paper, as to containment of ambiguity-limited classes, addresses
only $\fewp$.  (It is known that $\fewp$ is contained in the class
known as $\spp$ and is indeed so-called
$\spp$-low~\cite{koe-sch-tod-tor:j:few,fen-for-kur:j:gap,fen-for-li:j:gap-defin},
however that does not make our containments in restricted counting
classes uninteresting, as it seems unlikely that $\spp$ is contained
in \emph{any} restricted counting class, since SPP's ``no'' case
involves potentially exponential numbers of accepting paths, not zero
such paths.)
The
important paper of Cox and
Pay~\cite{cox-pay:t:wppplus}---along with many other
interesting results on counting classes---draws on the result of Borchert,
Hemaspaandra, and
Rothe~\cite{bor-hem-rot:j:powers-of-two} that appears as our
Theorem~\ref{t:bhr} to establish that
$\fewp \subseteq \rcb{2^{t}-1 \condition t \in
  \naturalnumberpositive}$ (note that the right-hand side is the
restricted counting class defined by the Mersenne numbers), a result
that itself implies $\fewp \subseteq \rcb{1,3,5,\dots}$.

``RC'' (restricted counting)
classes~\cite{bor-hem-rot:j:powers-of-two} are central to this paper.
The literature's earlier ``CP''
classes~\cite{cai-gun-har-hem-sew-wag-wec:j:bh2} might at first seem
similar, but they don't restrict rejection to the case of having zero
accepting paths.  Leaf languages~\cite{bov-cre-sil:j:sparse}, a
different framework, do have flexibility to express ``RC'' classes,
and so are an alternate notation one could use, though in some sense
they would be overkill as a framework here due to their extreme
descriptive power.  The class $\rcb{1,3,5,\dots}$ first appeared in
the literature under the name
$\rm
ModZ_{2}P$~\cite{bei-gil-her:cOUTbybei-gil-jour-and-her-jour:mod}.
Ambiguity-limited classes are also quite central to this paper, and
among those we study (see Section~\ref{s:defs}) are ones
defined, or given their notation that we use, in the following papers:
\cite{val:j:checking,bei:c:up1,wat:j:hardness-one-way,all-rub:j:print,lan-ros:c:up-circuit-and-hierarchy}.

The counting classes studied in this paper are all \emph{language}
classes, although each is or can be defined via $\sharpp$
functions.  (For example,
FewP is the class of all sets $L$ such that, for some polynomial
$p$ and some $\sharpp$ function $f$, we have that for each $x\in\sigmastar$
it holds that
(a)~$x \not\in L \implies f(x)=0$, and (b)~$x \in L \implies
1 \leq f(x) \leq p(|x|)$. Note that there is a
restriction in play there, namely,
that the $\sharpp$ function underpinning $L$ will
on no input $x$ take on any value strictly greater than
$p(|x|)$.)
The direct study
of \emph{counting classes of functions}, and the properties
and interrelations of those classes,
is an active research area, although we mention that to the best
of our knowledge the results of the
present paper do
not follow from any currently known results in that area.  As
a pointer to some of the interesting current research
in that area, we mention the work of Antonopoulos,
Bakali, Chalki,
Pagourtzis, Pantavos, and Zachos~\cite{cha:thesis:counting-easy-decision,ant-bak-cha-pag-pan-zac:j:counting-with-easy-decision}.
Finally, interesting but
fundamentally different in flavor from the work of this paper is the
broad stream of work focused on
completely classifying the complexity of various
\emph{families of counting
problems},
see,
e.g.,~\cite{cai-che:counting-problems-volume-1}.

P-printability is due to Hartmanis and
Yesha~\cite{har-yes:j:computation}.
Allender~\cite{all:j:pseudorandom} established a sufficient condition,
which we will discuss later, for the existence of infinite, P-printable
subsets of the primes.
As discussed in the text right after Corollary~\ref{c:primes-oh-one} and in
Footnote~\ref{f:tao-gaps}, none of the results of Ford, Maynard, Tao,
and others~\cite{for-gre-kon-tao:j:gaps2016,%
  may:j:prime-gaps,for-gre-kon-may-tao:j:gaps2018} about ``infinitely
often'' lower bounds on gaps in the primes, nor any possible future
bounds, can possibly be strong enough to be the sole obstacle to a
$\fewp \subseteq \rcnb{\primes}$ construction.

\section{Gaps, Ambiguity, and Iterative
  Constant-Setting}\label{sec:gaps}

What is the power of NPTMs whose number of accepting paths is 0 for
each string not in the set and is a prime for each string in the set?
In particular, does that class, $\rcnb{\primes}$, contain
$\fewp$ or, for that matter, any interesting ambiguity-limited
nondeterministic class?  That is the question that motivated this
work.

Why might one hope that $\rcnb{\primes}$ might contain some
ambiguity-limited classes?  Well, we clearly have that
$\np\subseteq \rcnb{\composites}$, so having the composites as our
acceptance targets allows us to capture all of $\np$. Why?  For any NP
machine $N$, we can make a new machine $N'$ that mimics $N$, except it
clones each accepting path into
four
accepting paths, and so when $N$
has zero accepting paths $N'$ has zero accepting paths, and when $N$ has at
least one accepting path $N'$ has a composite number of accepting
paths.\footnote{One
  should not think that the fact that
  $\np\subseteq \rcnb{\composites}$ (equivalently, $\np = \rcnb{\composites}$) holds means that
  $\conp\subseteq \rcnb{\primes}$; the latter in fact would
  immediately imply that $\np=\conp$, since
  for each $S \subseteq \naturalnumberpositive$ it holds that
  $\rcnb{S} \subseteq \np$.  As to whether the fact that
  $\np\subseteq \rcnb{\composites}$ holds means that
  $\conp\subseteq \rcnb{\naturalnumber - \composites}$ holds, the
  latter is not even well-defined, since the RC classes are defined
  only for sets $S$ satisfying $S \subseteq \naturalnumberpositive$.
  But even if one removes 0, and asks about
  $\conp\subseteq \rcnb{\naturalnumberpositive - \,\composites}$, for
  the same reason just mentioned that containment would imply
  $\np=\conp$.}

On the other hand, why might one suspect that interesting
ambiguity-limited nondeterministic classes such as $\fewp$ might
\emph{not} be contained in $\rcnb{\primes}$? Well, it is not even
clear that $\fewp$ is contained in the class of sets that are accepted
by NPTMs that accept via having a prime number of accepting paths, and
reject by having a nonprime number of accepting paths (rather than
being restricted to rejecting only by having zero accepting paths, as
is
$\rcnb{\primes}$).  That is, even a seemingly vastly more flexible
counting class does not seem to in any obvious way contain $\fewp$.

This led us to revisit the issue of identifying the sets
$S \subseteq \naturalnumberpositive$ that satisfy
$\fewp \subseteq \rcnb{S}$, studied previously by,
for example,
Borchert,
Hemaspaandra, and
Rothe~\cite{bor-hem-rot:j:powers-of-two}
and Cox and Pay~\cite{cox-pay:t:wppplus}.
In particular, Borchert, Hemaspaandra, and Rothe showed, by the
iterative constant-setting technique, the following theorem.
\begin{theorem}[{\cite[Theorem
    3.4]{bor-hem-rot:j:powers-of-two}}]\label{t:bhr}
  If $T \subseteq \naturalnumberpositive$ has an
  $(n+\bigo(1))$-nongappy, P-printable subset, then
  $\fewp \subseteq \rcnb{T}$.
\end{theorem}

From this, we immediately have the following corollary.
\begin{corollary}\label{c:primes-oh-one}
  If $\primes$ contains an $(n+ \bigo(1))$-nongappy, P-printable
  subset, then $\fewp \subseteq \rcnb{\primes}$.
\end{corollary}
Does $\primes$ contain an $(n+\bigo(1))$-nongappy, P-printable subset?
The Bertrand--Chebyshev
Theorem~\cite{che:j:bertrand-chebyshev-theorem} states that for each
natural number $k > 3$, there exists a prime $p$ such that $k<p<2k-2$.
Thus $\primes$ clearly has an $(n+\bigo(1))$-nongappy
subset.\footnote{%
  \label{f:tao-gaps}%
  We mention in passing that it follows from the fact that $\primes$
  clearly \emph{does} have an $(n+\bigo(1))$-nongappy subset that none
  of the powerful results by Ford, Maynard, Tao, and
  others~\cite{for-gre-kon-tao:j:gaps2016,%
    may:j:prime-gaps,for-gre-kon-may-tao:j:gaps2018} about
  ``infinitely often'' lower bounds for gaps in the primes, or in fact
  any results purely about lower bounds on gaps in the primes, can
  possibly prevent there from being a set of primes whose gaps are
  small enough that the set could, if sufficiently accessible, be used
  in a Cai--Hemachandra-type %
  iterative
  constant-setting construction
  seeking to show that $\fewp \subseteq \rcnb{\primes}$.  (In
  fact---keeping in mind that the difference between the value of a number
  and its coded length is exponential---the best such gaps known are
  almost exponentially too weak to preclude a Cai--Hemachandra-type
  iterative constant-setting construction.)  Rather, the only obstacle
  will be the issue of whether there is such a set that in addition is
  computationally easily accessible/thin-able, i.e., whether there is
  such an $(n+\bigo(1))$-nongappy subset of the primes that is
  P-printable.}
Indeed, since---with $p_i$ denoting the $i$th
prime---$(\forall \epsilon > 0)(\exists N)(\forall n>N) [p_{n+1}-p_n <
(p_n)^{\frac{3}{4} + \epsilon}]$~\cite{tch:j:prime-gaps-quite-tight},
it certainly holds that
represented in binary there
are primes at all but a finite number of bit-lengths.  Unfortunately,
to the best of our knowledge it remains an open research issue whether
there exists \emph{any} infinite, P-printable subset of the primes,
much less one that in addition is $(n+\bigo(1))$-nongappy.

In fact, the best sufficient condition we know of for the existence of
an infinite, P-printable set of primes is a relatively strong
hypothesis of Allender~\cite[Corollary 32 and the comment following
it]{all:j:pseudorandom} about the probabilistic complexity class
$\mathrm{RP}$~\cite{gil:j:prob-tms} and the existence of secure extenders.
However, that result does not promise that the infinite, P-printable
set of primes is $(n+\bigo(1))$-nongappy---not even now, when it is
known that primality is not merely in the class $\mathrm{RP}$ but even is in the
class $\p$~\cite{agr-kay-sax:j:primality}.

So the natural question to ask is: Can we at least lower the bar for
what strength of advance---regarding the existence of P-printable sets
of primes and the nongappiness of such sets---%
would suffice to allow $\rcnb{\primes}$ to contain some interesting
ambiguity-limited class?

In particular, the notion of nongappiness used in Theorem~\ref{t:bhr}
above means that our length gaps between adjacent elements of our
P-printable set must be bounded by an additive constant.  Can we
weaken that to allow larger gaps, e.g., gaps of multiplicative
constants, and still have containment for some interesting
ambiguity-limited class?

We show that the answer is yes. More generally, we show that there is
a tension and trade-off between gaps and ambiguity.  As we increase the
size of gaps we are willing to tolerate, we can prove containment
results for restrictive counting classes, but of increasingly small
levels of ambiguity.
On the other hand, as we lower the size of the gaps we are willing to
tolerate, we increase the amount of ambiguity we can handle.

It is easy to see that the case of constant-ambiguity nondeterminism
is so extreme that the iterative constant-setting method works for all
infinite sets regardless of how nongappy they are.  (It is even true
that the containment $\upleq{k} \subseteq \rcnb{T}$ holds for some finite sets $T$, such as
$\{1,2,3,\dots,k\}$;
but our point here is that it
holds for \emph{all} infinite sets
$T\subseteq\naturalnumberpositive$.)

\begin{theorem}\label{t:up-to-k}
  For each infinite set $T \subseteq \naturalnumberpositive$ and for
  each natural $k \geq1$, %
  $\upleq{k} \subseteq \rcnb{T}$.
\end{theorem}
Theorem~\ref{t:up-to-k} should be compared with the
discussion by %
Hemaspaandra and Rothe~\cite[p.~210]{hem-rot:j:rice} of
an NP-many-one-hardness result of Borchert and
Stephan~\cite{bor-ste:j:rice} and a $\upleq{k}$-1-truth-table-hardness
result.
In
particular, both those results are in the \emph{un}restricted setting,
and so neither implies Theorem~\ref{t:up-to-k}.

The proof of Theorem~\ref{t:up-to-k} is in
Appendix~A\@.
However, we recommend that the reader read
it, if at all, only after reading the proof of Theorem~\ref{t:meta1},
whose proof also uses (and within this paper, is the first
presentation of) iterative constant-setting, and is a more
interesting use of that approach.

\begin{corollary}\label{c:infinite-big-o}
  For each infinite set $T \subseteq \naturalnumberpositive$,
  $\upleq{\bigo(1)} \subseteq \rcnb{T}$.
\end{corollary}

So constant-ambiguity nondeterminism can be done by the restrictive
counting class based on the primes (as Corollary~\ref{c:infinite-big-o} immediately yields $\upleq{\bigo(1)} \subseteq \rcnb{\primes})$.
However, what we are truly
interested in is whether we can achieve a containment for
superconstant levels of ambiguity.  We in fact
can do so, and we now present
such results for
a
range of cases between constant ambiguity ($\upleq{\bigo(1)}$) and
polynomial ambiguity ($\fewp$).
Just as Corollary~\ref{c:primes-oh-one} follows from Theorem~\ref{t:bhr}, so also do each of our Theorems~\ref{t:kn} and \ref{t:tradeoff}, Parts~1--3 each have the obvious analogous corollary regarding $\rcnb{\primes}$.

We first define a broader notion of nongappiness.
\begin{definition}\label{d:nongappy-main}
  Let $F$ be any function mapping $\realnumberpositive$ to
  $\realnumberpositive$.
A set $S \subseteq \naturalnumberpositive$
  is $F$-nongappy if $S \neq \emptyset$ and
  $(\forall m\in S)(\exists m'\in S) [m' > m \land |m'| \leq F(|m|)]$.\footnote{In two later definitions, \ref{d:4-8} and~\ref{def:polylog-nongappy},
we apply Definition~\ref{d:nongappy-main}
to classes of functions.  In each case,
we will directly define that, but in fact will do so as the natural
lifting (namely, saying a set is $\calf$-nongappy exactly if there is an
$F \in \calf$ such that the set is $F$-nongappy).  The reason we do not
directly define lifting as applying to all classes $\calf$ is in small
part
that we need it only in those two definitions, and in large part because
doing so could cause confusion, since an
earlier definition (Def.~\ref{def: n+O(1)-nongappy}) that is connecting
to earlier work is using as a syntactic notation an
expression that itself would be caught up by such a lifting (though
the definition given in
Def.~\ref{def: n+O(1)-nongappy}
is consistent with the lifting reading,
give or take the fact that we've now broadened
our focus to the reals rather than the
naturals).}
\end{definition}

This definition sets $F$'s domain and codomain to include real numbers,
despite the fact that the underlying $F$-nongappy set $S$
is of the type $S\subseteq \naturalnumberpositive$.
The codomain is set to include real numbers because many notions of nongappiness we examine rely on non-integer values.
Since we are often iterating functions, we thus set $F$'s domain to be real numbers as well.
Doing so does not cause problems as to computability because $F$ is a function
that is never actually computed by the Turing machines in our proofs; it is merely one
that is mathematically reasoned about in the analysis of the nongappiness of sets
underpinning restricted counting classes.

The following theorem generalizes the iterative constant-setting
technique that Borchert, Hemaspaandra, and Rothe used to prove
Theorem~\ref{t:bhr}.
\begin{theorem}\label{t:meta1}
  Let $F$ be a function mapping from $\realnumberpositive$ to
  $\realnumberpositive$ and let $n_0$ be a positive natural number
  such that $F$ restricted to the domain
  $\{t \in \realnumberpositive \condition t \geq n_0\}$ is
  nondecreasing and for all $t \geq n_0$ we
  have %
  (a)~$F(t) \geq t + 2$ and
  (b)~$(\forall c \in \naturalnumberpositive)[cF(t) \geq F(ct)]$.
  Let $j$ be a function, mapping from $\naturalnumber$ to
  $\naturalnumberpositive$,
that is at most polynomial in the value of its
  input and is computable in time polynomial in the value of its
  input.
  Suppose $T \subseteq \naturalnumberpositive$
  has an $F$-nongappy, P-printable subset $S$. Let $\lambda = 4 + |s|$
  where $s$ is the smallest element of $S$ with $|s| \geq n_0$.
  If for some $\beta \in \naturalnumberpositive$,
  ${F^{[j(n)]}(\lambda)} = \bigo(n^\beta)$, then
  $\upleq{j(n)} \subseteq \rcnb{T}$.
\end{theorem}

This theorem has a nice interpretation: a sufficient condition for an
ambiguity-limited class $\upleq{j(n)}$ to be contained in a particular
restricted counting class is for there to be at least $j(n)$
elements that are reachable in polynomial time in an
$F$-nongappy subset
of the set that defines the counting class, assuming that the
nongappiness of the counting class and the ambiguity of the
$\upleq{j(n)}$ class satisfy the above conditions.

\begin{proof}[Proof of Theorem~\ref{t:meta1}]
Let $F$, $j$, $n_0$, $T$, and $S$ be as per the theorem statement.
Suppose $(\exists \beta' \in \naturalnumberpositive)[{F^{[j(n)]}(\lambda)} = \bigo(n^{\beta'})]$,
and fix a value  $\beta \in \naturalnumberpositive$ such that
$F^{[j(n)]}(\lambda) = \bigo(n^{\beta})$.

We start our proof by defining three sequences of constants that will be
central in our iterative constant-setting argument,
and giving bounds on their growth.
Set $c_1$ to be the least element of $S$ with $|c_1| \geq n_0$.
For $n \in \{2, 3, \ldots\}$, given $c_1, c_2, \ldots, c_{n-1}$, we set
\begin{equation}
    \label{eq:b_n}
    b_n = \sum_{1 \leq \ell \leq n-1} c_\ell {n \choose \ell}.
\end{equation}
With $b_n$ set, we define $a_n$ to be the least element of $S$
such that $a_n > b_n$. Finally, we set $c_n = a_n - b_n$.
We now show that $\max_{2 \leq \ell \leq j(n)} |a_\ell|$ and $\max_{1 \leq \ell \leq j(n)} |c_\ell|$ are both at most polynomial in $n$.
Take any $i \in \{2, 3, \ldots\}$.
By the construction above and since $S$ is $F$-nongappy,
we have $|c_i| \leq |a_i| \leq F(|b_i|)$.
Using our definition of $b_i$ from Eq.~\ref{eq:b_n} we
get $b_i = \sum_{1 \leq k \leq i-1} c_k {i \choose k} \leq
(i-1)(\max_{1 \leq k \leq i-1} c_k){i \choose \lceil \frac{i}{2} \rceil} \leq
(\max_{1 \leq k \leq i-1} c_k)(2^{2i})$.
Thus we can bound the length of $b_i$ by
$|b_i| \leq 2i + \max_{1 \leq k \leq i-1} |c_k|$.
Since this is true for all $i \in \{2, 3, \ldots\}$,
it follows that if $\max_{1 \leq \ell \leq j(n)} |c_\ell|$ is at most polynomial
in $n$, then $\max_{2 \leq \ell \leq j(n)} |b_\ell|$ is at most polynomial in $n$, and
since for all $i$, $a_i = b_i + c_i$, $\max_{2 \leq \ell \leq j(n)} |a_\ell|$ is at most polynomial in $n$.

We now show that $\max_{1 \leq \ell \leq j(n)} |c_\ell|$ is in fact polynomial in $n$ via the following claim, which we prove by induction: for all $i \in \{2, 3, \ldots\}$,
$\max_{1 \leq \ell \leq i} |c_\ell| \leq (i-1) F^{[(i-1)]}(\lambda)$.
We showed above that for any $i \in \{2, 3,
\ldots\}$, $|c_i| \leq F(|b_i|)$ and $|b_i| \leq 2i + \max_{1 \leq k \leq i-1}
|c_k|$. For each $i \in \{2, 3, \ldots\}$, we have $|b_i| \geq |c_1| \geq n_0$.
Since $F$ restricted to $\{t \in \realnumberpositive \condition t \geq n_0\}$ is nondecreasing, for all $i \in \{2, 3, \ldots, n\}$,
\begin{equation} \label{eq:c_i_bound}
    |c_i| \leq F(|b_i|) \leq F(2i + \max_{1 \leq k \leq i-1} |c_k|).
\end{equation}
Recall that $\lambda = 4 + |c_1|$.  For the base case of our induction, notice
that substituting $i = 2$ into Eq.~\ref{eq:c_i_bound} gives $|c_2| \leq F(4 +
|c_1|) = F(\lambda)$.
Also, by condition (a) and the fact that $|c_1| \geq n_0$ we have $|c_1| < F(|c_1|) \leq F(4 + |c_1|) = F(\lambda)$.
Thus $\max_{1 \leq \ell \leq 2} |c_\ell| \leq F(\lambda)$, which is the claimed inequality for $i = 2$.
Suppose, for induction, that the claim holds for some $i \geq 2$.
Since $|c_1| \geq n_0$, condition (a) and the fact that $F$ restricted to $\{t \in \realnumberpositive \condition t \geq n_0\}$ is nondecreasing give us $F^{[i-1]}(\lambda) = F^{[i-1]}(4+|c_1|) \geq 2(i-1) + 4 + |c_1| \geq 2(i+1)$.
Plugging $i+1$ into Eq.~\ref{eq:c_i_bound} and using the fact that all the inputs to $F$ involved are greater than $n_0$ (and so are in the domain where $F$ is nondecreasing), we have
\begin{align*}
    |c_{i+1}| &\leq F(2(i+1) + \max_{1 \leq \ell \leq i} |c_\ell|) \\
    &\leq F(2(i+1) + (i-1)F^{[i-1]}(\lambda)) \tag{by the inductive hypothesis} \\
    &\leq F(F^{[i-1]}(\lambda) + (i-1)F^{[i-1]}(\lambda)) \\
    &= F(iF^{[i-1]}(\lambda)) \\
    &\leq iF^{[i]}(\lambda) \tag{by condition (b)}.
\end{align*}
Since $\max_{1\leq \ell \leq i} |c_\ell| \leq (i-1)F^{[i-1]}(\lambda) \leq i F^{[i]}(\lambda)$, we have $\max_{1 \leq \ell \leq i+1} |c_\ell| \leq i F^{[i]}(\lambda)$, which is
the claimed inequality for $i+1$. By induction,
the claim holds.

Substituting $i = j(n)$ into the inequality we just proved, we get $\max_{1 \leq \ell \leq j(n)} |c_\ell| \leq j(n)F^{[j(n)]}(\lambda)$.
Since $F^{[j(n)]}(\lambda) = \bigo(n^\beta)$ and by hypothesis $j(n)$ is at most polynomial in $n$, $j(n)F^{[j(n)]}(\lambda)$ is at most polynomial in $n$. Hence $\max_{1 \leq \ell \leq j(n)} |c_\ell|$ is at most polynomial in $n$.

We now proceed to show that $\upleq{j(n)} \subseteq \rcnb{T}$.
Let $L$ be a language in $\upleq{j(n)}$, witnessed by an NPTM $\hat{N}$.
To show $L \in \rcnb{T}$ we give a description of an NPTM $N$ that,
on each input $x$, has 0 accepting paths if $x \notin L$,
and has $\acc_N(x) \in T$ if $x \in L$.
On input $x$, our machine $N$ computes $j(|x|)$ and then computes the constants $c_1, c_2, \ldots, c_{j(|x|)}$ as described above.
Note that this computation relies on the P-printability of $S$, which ensures that the constants $a_i$ (which must be computed to compute $c_i$) are computable.
Then $N$ nondeterministically guesses an integer $i \in \{1, 2, \ldots, j(|x|)\}$,
and nondeterministically guesses a cardinality-$i$ set of paths of $\hat{N}(x)$.
If all the paths guessed in a cardinality-$i$ set are accepting paths, then $N$ branches into $c_i$ accepting paths;
otherwise, that branch of $N$ rejects.
Of course, if $\hat{N}(x)$ has fewer than $i$ paths, then the subtree of $N$ that guessed
$i$ will have zero accepting paths, since we cannot guess $i$ distinct paths of $\hat{N}(x)$. We claim that $N$ shows $L \in \rcnb{T}$.

Consider any input $x$. If $x \notin L$, then clearly for all $i \in \{1,2,\ldots,j(|x|)\}$ each cardinality-$i$ set of paths of $\hat N$ guessed will have at least one rejecting path,
and so $N$ will have no accepting path.
Suppose $x \in L$. Then $\hat{N}$ must have some number of accepting paths $k$.
Since $\hat{N}$ witnesses $L \in \upleq{j(n)}$, we must have $1 \leq k \leq j(|x|)$.
Our machine $N$ will have $c_1$ accepting paths for each accepting path of $\hat N$, $c_2$ additional accepting paths for each pair of accepting paths of $\hat N$, $c_3$ additional accepting paths for each triple of accepting paths of $\hat N$, and so on.
Of course, for any cardinality-$i$ set where $i > k$, at least one of the paths must be rejecting, and so $N$ will have no accepting paths from guessing each $i > k$.
Thus we have $\acc_N(x) = \sum_{1 \leq \ell \leq k} c_\ell {k \choose \ell}$.
If $k = 1$, we have $\acc_N(x) = c_1$. If $2 \leq k \leq j(|x|)$, then
$\acc_N(x) = c_k + \sum_{1 \leq \ell \leq k-1} c_\ell {k \choose \ell} = c_k + b_k = a_k$.
In either case, $\acc_N(x) \in S$, and hence $\acc_N(x) \in T$.
To complete our proof for $L \in \rcnb{T}$ we need to check that $N$ is an NPTM.

Note that, by assumption, $j(|x|)$ can be computed in time polynomial in $|x|$.
Furthermore, the value $j(|x|)$ is at most polynomial in $|x|$, and so $N$'s simulation of each cardinality-$i$ set of paths of $\hat N$ can be done in time polynomial in $|x|$.
Since $S$ is P-printable and $\max_{1 \leq i \leq j(|x|)} |a_i|$ is at most polynomial in $|x|$, finding the constants $a_i$ can be done in time polynomial in $|x|$.
Also, since $\max_{1 \leq i \leq j(|x|)} |c_i|$ is at most polynomial in $|x|$, the addition and multiplication to compute each $c_i$ can be done in time polynomial in $|x|$.
All other operations done by $N$ are also polynomial-time, and so $N$ is an NPTM.
\end{proof}

It is worth noting that in general iterative constant-setting proofs it is
sometimes useful to have a nonzero constant $c_0$
in order to add a constant number $c_0{i \choose 0} = c_0$ of accepting paths.
However, when trying to show
containment in a restricted counting class (as is the case here), we set
$c_0 = 0$ to ensure that $\acc_N(x) = 0$ if $x \notin L$, and so we
do not even have a $c_0$ but rather start iterative
constant-setting and its sums with the
$c_1$ case (as in Eq.~\ref{eq:b_n}).

Theorem~\ref{t:meta1} can be applied to get
complexity-class
containments.
In particular, we now define a notion of nongappiness based on a
multiplicative-constant increase in lengths, and we show---as
Theorem~\ref{t:kn}---that this
notion of nongappiness allows us to accept all sets of logarithmic
ambiguity.

\begin{definition}\label{d:4-8}
  A set $S \subseteq \naturalnumberpositive$ is $\bigo(n)$-nongappy if
  $S\neq \emptyset$ and
  $(\exists f \in \bigo(n))(\forall m\in S)(\exists m'\in S) [m'>m
  \land |m'| \leq f(|m|)]$.
\end{definition}

As Fact~\ref{f:O(n)-ng-equiv}
we note that
one can view this definition in a form similar to Borchert,
Hemaspaandra, and Rothe's definition to see that $\bigo(n)$-nongappy
sets are, as to the increase in the lengths of consecutive elements,
bounded by a multiplicative constant. (In terms of values, this means
that the gaps between the values of one element of the set and the
next are bounded by a polynomial increase.)

\begin{fact} \label{f:O(n)-ng-equiv} A set
  $S \subseteq \naturalnumberpositive$ is $\bigo(n)$-nongappy if and
  only if there exists $k \in \naturalnumberpositive$ such that $S$ is $kn$-nongappy.
\end{fact}
\begin{proof}
  As to the ``only if'' direction, suppose $S$ is $\bigo(n)$-nongappy.
  By definition of $\bigo(n)$-nongappy, $S \neq \emptyset$ and there is a function $f \in \bigo(n)$
  such that $(\forall m \in S)(\exists m' \in S)[m' > m \land |m'| \leq f(|m|)]$.
  Since $f$ is $\bigo(n)$, we can find constants $k_0$ and $n_0$ in $\naturalnumberpositive$
  such that $(\forall t \geq n_0)[f(t) \leq k_0 t]$. Let $k = \max(f(1), f(2), \ldots, f(n_0), k_0)$.
  Then for all $n \in \naturalnumberpositive$, $f(n) \leq kn$.
(Though $f$ has domain $\realnumberpositive$, it is enough to have this bound
  just for the positive natural numbers since our definition of nongappy only invokes the function on positive naturals.)
  It is easy to see that this $k$ is such that $S$ is $kn$-nongappy.

  As to the ``if'' direction,
  $kn$ is certainly $\bigo(n)$.
\end{proof}%

\begin{theorem}
  \label{t:kn}
  If $T \subseteq \naturalnumberpositive$
  has an
  $\bigo(n)$-nongappy, P-printable subset, then
  $\upleq{\bigo(\log n)} \subseteq \rcnb{T}$.
\end{theorem}
\begin{proof}
  Suppose $T \subseteq \naturalnumberpositive$ has an $\bigo(n)$-nongappy, P-printable subset.
  By the ``only if'' direction of
  Fact~\ref{f:O(n)-ng-equiv},
  there exists a $k \in \naturalnumberpositive$ such that $T$ has a
  $kn$-nongappy, P-printable subset.  We can assume $k \geq 2$ since
  if a set has a $1n$-nongappy, P-printable subset then it also has a
  $2n$-nongappy, P-printable subset.
  Let $F: \realnumberpositive \to \realnumberpositive$ be the function $F(t) = kt$.
  The function $F$
  satisfies the conditions from %
  Theorem~\ref{t:meta1} since for all $t \geq 2$,
  $F(t) = kt \geq t+2$, $(\forall c \in \naturalnumberpositive)[cF(n) = ckn = F(cn)]$, and $F$ is
  nondecreasing on $\realnumberpositive$. Let $\lambda = 4+|s|$ where $s$ is the smallest element of the $kn$-nongappy, P-printable subset of $T$ such that the conditions on $F$ hold for all $t \geq |s|$, i.e.,
  $s$ is the smallest element of the $kn$-nongappy, P-printable subset of $T$ such that $|s| \geq 2$.
For any function $g: \naturalnumber \to \realnumberatleastone$
    satisfying
$g(n) = \bigo(\log n)$ it is not
hard to see (since for each
natural $n$ it holds that $\log(n+2)\geq 1$) that there must exist some
$d \in \naturalnumberpositive$ such that
$(\forall n\in \naturalnumberpositive)[g(n)\leq d \log(n+2)]$,
and
hence
$\upleq{g(n)} \subseteq \upleq{d \log(n+2)} = \upleq{\floor{d \log(n+2)}}$.  Additionally,
  $j(n) = \floor{d \log(n+2)}$ satisfies the conditions from
  Theorem~\ref{t:meta1} since $j(n)$
can be
  computed in time polynomial in $n$
  (for example by doing a linear search for the largest $i \in \naturalnumber$ such that $2^i \leq (n+2)^d$)
  and has value at most polynomial
  in $n$.
Applying Theorem~\ref{t:meta1}, to prove that
  $\upleq{j(n)} \subseteq \rcnb{T}$ it suffices to show that there is
  some $\beta \in \naturalnumberpositive$ such that
  ${F^{[j(n)]}(\lambda)} = \bigo(n^\beta)$.
So it suffices to
  show that for some $\beta \in \naturalnumberpositive$ and for sufficiently large $n$, ${F^{[j(n)]}(\lambda)} \leq n^\beta$. Note
  that $F^{[j(n)]}(\lambda) = k^{j(n)}\lambda$. So it is enough to
  show that there exists $\beta$ such that for sufficiently large $n$,
  $k^{j(n)}\lambda \leq n^\beta$, or (taking logs)
  equivalently that for large $n$,
  $\floor{d \log(n+2)} \log k + \log \lambda \leq \beta \log n$.
  The left-hand side of this inequality is at most $d\log(n+2)\log k + \log \lambda \leq \log(n)[2d\log k + \log\lambda]$ which, for all $\beta \geq 2d\log k + \log \lambda$, is at most $\beta \log n$.
  Thus there exits a constant $\beta$ such that $F^{[j(n)]}(\lambda) = \bigo(n^\beta)$.
Hence for any function $g: \naturalnumber \to \realnumberatleastone$ satisfying
$g(n) = \bigo(\log n)$
we have
  that there exists a function $j$
such
  that $\upleq{g(n)} \subseteq \upleq{j(n)} \subseteq \rcnb{T}$.
\end{proof}

In order for the iterative constant-setting approach used in
Theorem~\ref{t:meta1} to be applicable, it is clear that we need to
consider $\up$ classes that have at most polynomial ambiguity, because
otherwise the
constructed NPTMs could not guess large enough collections of paths
within polynomial time. Since in the statement
of Theorem~\ref{t:meta1} we use the function
$j$ to denote the ambiguity of a particular $\up$ class, this requires
$j$ to be at most polynomial in the value of its input. Furthermore,
since our iterative constant-setting requires having a bound on the
number of accepting paths the $\up$ machine could have had on a
particular string, we also need to be able to compute the function $j$
in time polynomial in the value of its input. Thus the limitations on
the function $j$ are natural and seem difficult to remove.
Theorem~\ref{t:meta1} is flexible enough to, by a
proof similar to that of Theorem~\ref{t:kn}, imply Borchert,
Hemaspaandra, and Rothe's result stated in Theorem~\ref{t:bhr} where
$j$ reaches its polynomial bound.
Another limitation of Theorem~\ref{t:meta1} is that it requires
that for all $t$ greater than or equal to a fixed constant $n_0$,
($\forall c \in \naturalnumberpositive)[cF(t) \geq F(ct)]$.
It is
possible to prove a similar result where
for all $t$ greater than or equal to a fixed constant $n_0$,
$(\forall c \in \realnumberatleastone)[cF(t) \leq F(ct)]$,
which
we now do as Theorem~\ref{t:meta2}.
For
each of
the $F$-nongappy
set classes that seemed most interesting to us,
one of these two conditions turned out to hold for $F$, and so one of
the two results shown in Theorems~\ref{t:meta1} and~\ref{t:meta2} was
applicable in finding the ambiguity-limited class that is
contained in a restricted counting class associated with a set of
natural numbers with some $F$-nongappy, P-printable subset.

\begin{theorem}\label{t:meta2}
  Let $F$ be a function mapping from $\realnumberpositive$ to
  $\realnumberpositive$ and let $n_0$ be a positive natural number
  such that $F$ restricted to the domain
  $\{t \in \realnumberpositive \condition t \geq n_0\}$ is
  nondecreasing and for all $t \geq n_0$ we have (a)~$F(t) \geq t + 2$
  and (b)~$(\forall c \in \realnumberatleastone)[cF(t) \leq F(ct)]$.
  Let $j$ be a function mapping from $\naturalnumber$ to
  $\naturalnumberpositive$
that is computable in time polynomial in the value
  of its input and whose output is at most polynomial in the value of
  its input. Suppose $T \subseteq \naturalnumberpositive$
  has an $F$-nongappy, P-printable subset $S$.  Let
  $\lambda = 4 + |s|$ where $s$ is the smallest element of $S$ with
  $|s| \geq n_0$.
  If for some $\beta$, ${F^{[j(n)]}(j(n) \lambda)} = \bigo(n^\beta)$,
  then $\upleq{j(n)} \subseteq \rcnb{T}$.
\end{theorem}

How does this theorem compare with our other metatheorem,
Theorem~\ref{t:meta1}?
Since in both metatheorems $F$ is nondecreasing after a %
prefix,
speaking informally and broadly,
the functions $F$ where (after a %
prefix)
$(\forall c \in \realnumberatleastone)[cF(t) \leq F(ct)]$ holds
grow faster
than the functions $F$ where (after a
prefix)
$(\forall c \in \naturalnumberpositive)[cF(t) \geq F(ct)]$ holds.
(The examples we give of applying the two theorems reflect this.)
So, this second metatheorem is accommodating larger gaps
in the sets of integers
that define our restricted counting class, but is also
assuming a slightly
stronger condition for the containment of an ambiguity-limited class
to follow. More specifically, since we have the extra factor of $j(n)$
inside of the iterated application of $F$, we may need even more than
$j(|x|)$ elements to be reachable in polynomial time (exactly how many
more will depend on the particular function~$F$).

\begin{proof}[Proof of Theorem~\ref{t:meta2}]
The proof follows almost identically to the proof of Theorem~\ref{t:meta1}.
Let $F$, $j$, $n_0$, $T$, and $S$ be as per the theorem statement.
Suppose $(\exists \beta' \in \naturalnumberpositive)[F^{[j(n)]}(j(n)\lambda) = \bigo(n^{\beta'})]$, and let $\beta \in \naturalnumberpositive$ be some specific, fixed
$\beta'$ value instantiating that.
We define the sequences of constants $a_n$, $b_n$, and $c_n$ exactly as in the proof of Theorem~\ref{t:meta1}.
We now show that $\max_{2 \leq \ell \leq j(n)} |a_\ell|$
and $\max_{1 \leq \ell \leq j(n)} |c_\ell|$ are at most polynomial in $n$.
For the same reasons as in the proof of Theorem~\ref{t:meta1}, if $\max_{1 \leq \ell \leq j(n)} |c_\ell|$ is at most polynomial in $n$ then $\max_{2 \leq \ell \leq j(n)} |a_\ell|$ is at most polynomial in $n$.

We prove that $\max_{1 \leq \ell \leq j(n)} |c_\ell|$ is at most polynomial in $n$ by proving the following claim via induction: for all $i \in \{2, 3, \ldots\}$, $\max_{1 \leq \ell \leq i} |c_\ell| \leq F^{[i-1]}((i-1)\lambda)$.
Notice that Eq.~\ref{eq:c_i_bound} from the proof of Theorem~\ref{t:meta1}, which says that $|c_i| \leq F(2i + \max_{1 \leq k \leq i-1} |c_k|)$ still holds, since deriving it did not use any of the assumptions that are different in that theorem.
Plugging in $i = 2$ gives us $|c_2| \leq F(4 + |c_1|) = F(\lambda)$.
Also, by condition (a) and the fact that $|c_1| \geq n_0$, $|c_1| < F(|c_1|)$ which, since $F$ restricted on $\{t \in \realnumberpositive \condition t \geq n_0\}$ is nondecreasing, is at most $F(\lambda)$.
Thus $\max_{1 \leq k \leq 2} |c_k| \leq F(\lambda)$, which is the $i=2$ case of our claim.
Suppose now that the claim holds for some $i \geq 2$.
Since $\lambda \geq 4$, we have $2(i+1) \leq 2(i+1) + \lambda - 4$.
Condition (a) implies $2(i+1) \leq F^{[i-1]}(\lambda) = \frac{i-1}{i-1}F^{[i-1]}(\lambda)$.
Since it follows from condition (b) that for all $\ell \in \naturalnumberpositive$, $c \geq 1$, and $t \geq n_0$, $cF^{[\ell]}(t) \leq F^{[\ell]}(ct)$, we have $2(i+1) \leq \frac{1}{i-1}F^{[i-1]}((i-1)\lambda)$.
Plugging $i+1$ into Eq.~\ref{eq:c_i_bound} we have
\begin{align*}
    |c_{i+1}| &\leq F(2(i+1) + \max_{1 \leq k \leq i} |c_k|) \\
    &\leq F(2(i+1) + F^{[i-1]}((i-1)\lambda)) \tag{by the inductive hypothesis} \\
    &\leq F(\frac{1}{i-1}F^{[i-1]}((i-1)\lambda) + F^{[i-1]}((i-1)\lambda)) \\
    &= F(\frac{i}{i-1}F^{[i-1]}((i-1)\lambda)) \\
    &\leq F^{[i]}(i\lambda). \tag{by condition (b)}
\end{align*}
Since $\max_{1 \leq \ell \leq i} |c_\ell| \leq F^{[i-1]}((i-1)\lambda) \leq F^{[i]}(i\lambda)$, we have $\max_{1 \leq \ell \leq i+1} |c_\ell| \leq F^{[i]}(i\lambda)$,
which, by induction, proves the claim.
Plugging $i = j(n)$ into the claim we just proved we get that $\max_{1 \leq \ell \leq j(n)} |c_\ell| \leq F^{[j(n)-1]}((j(n)-1)\lambda) \leq F^{[j(n)]}(j(n)\lambda) = \bigo(n^\beta)$.

Consider any language $L \in \upleq{j(n)}$ witnessed by machine $\hat{N}$.
Given $\hat{N}$ and the sequences of constants we defined, we construct a machine $N$ identically to the proof of Theorem~\ref{t:meta1}.
By the arguments in the proof of that theorem, $N$ accepts $L$ in an $\rcnb{T}$-like fashion, apart from the fact that we have not yet shown $N$ to be an NPTM\@.  As to that final issue,
since we showed that $\max_{2 \leq i \leq j(n)} |a_i|$ and $\max_{1 \leq i \leq j(n)} |c_i|$ are at most polynomial in $n$,
the arguments in the proof of Theorem~\ref{t:meta1} for why $N$ is an NPTM still hold.
Thus $L \in \rcnb{T}$, which completes our proof.
\end{proof}

We now discuss some other notions of nongappiness and obtain
complexity-class containments regarding them using Theorem~\ref{t:meta2}.
Theorem~\ref{t:n^k} and its corollary, Corollary~\ref{c:lpw-containment}, were flawed in previous versions
of this paper; we thank an anonymous
ACM~TOCT referee for spotting the problem.

\begin{theorem}\label{t:n^k}
    For any number $k \in \realnumberpositive$ that can be expressed as $k = 2^{c/2}$ for some $c \in \naturalnumberpositive$, if $T \subseteq \naturalnumberpositive$ has an $n^k$-nongappy, P-printable subset, then
    $\upleq{\bigo(1) + \frac{\log\log n}{2\log k}} \subseteq \rcnb{T}$.
\end{theorem}
\begin{proof}
 Let $c$ be an arbitrary positive natural number, and let $k = 2^{c/2}$ (notice that $k \geq \sqrt{2}$).
Suppose $T \subseteq \naturalnumberpositive$ is a set having an $n^k$-nongappy, P-printable subset.  We argue that $\upleq{\bigo{(1)} + \frac{\log \log n}{2 \log k}} \subseteq
 \rcnb{T}$.

Set $F: \realnumberpositive \to \realnumberpositive$ to be $F(t) = t^k$.
$F$ satisfies the conditions on the $F$ in Theorem~\ref{t:meta2}
because $F$ is nondecreasing on $\realnumberpositive$ (since $k > 0$), and for all $t \geq 4$ we have
(a) $t^k - t = t(t^{k-1} - 1) \geq 4(4^{\sqrt{2}-1} - 1) > 2$,
which means $F(t) \geq t + 2$, and
(b) $(\forall c \in \realnumberatleastone)[cF(t) = ct^k \leq (ct)^k = F(ct)]$.
Let $\lambda = 4+|s|$ where $s$ is the smallest element of the $n^k$-nongappy, P-printable subset of $T$ where the conditions on $F$ hold for all $t \geq |s|$, i.e.,
$s$ is the smallest element of the $n^k$-nongappy, P-printable subset such that $|s| \geq 4$.
For every function $g: \naturalnumber \to \realnumberatleastone$
  satisfying $g(n)=\bigo{(1)} + \frac{\log \log n}{2 \log k}$
  there exists a $d \in \naturalnumberpositive$ such that
  $g(n)\leq \floor{d + \frac{\log \log n}{2 \log k}}$, and hence
  $\upleq{g(n)} \subseteq \upleq{\floor{d + \frac{\log \log n}{2 \log k}}}$.%
  \footnote{Note that the expression $d + \frac{\log\log n}{2 \log k}$ is not problematic despite the fact that 0 is a valid input since we have globally redefined $\log(\cdot)$ to mean $\log(\max(1, \cdot))$. The same is true for similar expressions that appear in the proof of Theorem~\ref{t:tradeoff}.}
  The function $j(n) = \floor{d + \frac{\log \log n}{2 \log k}} = \floor{d + \frac{\log\log n}{c}}$
  satisfies the conditions
  on $j$
  of Theorem~\ref{t:meta2}, since $j(n)$ can be
  computed in time polynomial in the value $n$
  (since $\floor{\frac{\log\log n}{c}}$ can be computed by doing a linear search for the largest natural number $i$ such that $2^{2^{ci}} \leq n$)
  and $j(n)$ has value at
  most polynomial in the value $n$.  Applying Theorem~\ref{t:meta2},
to prove that $\upleq{j(n)} \subseteq \rcnb{T}$ it suffices to show
that for some $\beta \in \naturalnumberpositive$ and for sufficiently large $n$, $F^{[j(n)]}(j(n) \lambda) \leq n^{\beta}$.
Plugging in
  $F^{[j(n)]}(j(n) \lambda) = (j(n) \lambda)^{k^{j(n)}}$ and taking logarithms,
  we see that this is the
  same as showing that there exists $\beta$ such that for sufficiently large $n$,
  $j(n) \log k + \log \log (j(n) \lambda) \leq \log \beta + \log \log n$.
  For large $n$,
  $j(n) \log k + \log \log (j(n) \lambda) \leq 2j(n) \log k \leq 2d\log k + \log\log n$.

  Setting $\beta = 2^{2d\log k}$ gives us the
  desired inequality.
Thus for any function $g: \naturalnumber \to \realnumberatleastone$
  satisfying $g(n)=\bigo{(1)} + \frac{\log \log n}{2 \log k}$,
  we have shown that there exists
a function $j$
  such that
  $\upleq{g(n)} \subseteq \upleq{j(n)} \subseteq \rcnb{T}$ (and
  so we have
 $\upleq{g(n)} \subseteq \rcnb{T}$).
\end{proof}

Theorem~\ref{t:n^k} has an interesting consequence when applied to the
Mersenne primes.  In particular, as we now show, it can be used to
prove that
the Lenstra--Pomerance--Wagstaff Conjecture implies that the
$(\bigoh(1) + \log\log n)$-ambiguity sets in NP each belong to
$\rcnb{\primes}$.

A Mersenne prime is a prime of the form $2^k-1$. We will use the
Mersenne prime counting function $\mu(n)$
to denote the number of Mersenne primes with length less than or equal
to~$n$ (when represented in binary).  The Lenstra--Pomerance--Wagstaff
Conjecture~\cite{pom:j:primality-testing,wag:j:mersenne-numbers}
(see also~\cite{cal:url:mersenne-heuristic}) asserts that there are
infinitely many Mersenne primes, and that $\mu(n)$ grows
asymptotically as $e^\gamma \log n$ where $\gamma \approx 0.577$ is
the Euler--Mascheroni constant.
(Note: We say that $f(n)$ grows asymptotically as $g(n)$ when
  $\lim_{n\to\infty} f(n)/g(n) = 1$.)

Having infinitely many Mersenne primes immediately yields an infinite
P-printable subset of the primes.
In particular, on input $1^n$ we can print all Mersenne primes of
length less than or equal to $n$ in polynomial time by just checking
(using a deterministic polynomial-time primality
test~\cite{agr-kay-sax:j:primality}) each number of the form $2^k-1$
whose length is less than or equal to $n$, and if it is prime then
printing it.

If the Lenstra--Pomerance--Wagstaff Conjecture holds, what can we also
say about the gaps in the Mersenne primes?  We address that with the
following result.

\begin{theorem}\label{t:lpw}
  If the Lenstra--Pomerance--Wagstaff Conjecture holds, then for each
  $\epsilon > 0$ the primes (indeed, even the Mersenne primes) have an
  $n^{1 + \epsilon}$-nongappy, P-printable subset.
\end{theorem}

\begin{proof}
  Assuming the Lenstra--Pomerance--Wagstaff Conjecture, there must be
  an infinite number of Mersenne primes. For each
  $i \in \naturalnumberpositive$, let $M_i$ denote the $i$th Mersenne
  prime.

  The
  density assertion of the
  Lenstra--Pomerance--Wagstaff Conjecture %
  implies that
  $(\forall \delta >0) (\exists N(\delta) \in \naturalnumberpositive)
  (\forall n>N(\delta)) [(1-\delta)(e^\gamma \log n) \leq \mu(n) \leq
  (1+\delta)(e^\gamma \log n)]$.  Suppose, by way of seeking a
  contradiction, that for some $\epsilon > 0$ there are infinitely
  many $n$ such that for two successive Mersenne primes $M_n$ and
  $M_{n+1}$, %
  $|M_{n+1}| > |M_n|^{1+\epsilon}$.  Fix a $\delta$ satisfying
  $\delta < \frac{\epsilon}{2(\epsilon + 2)}$, and let $M_n$ and
  $M_{n+1}$ be two consecutive Mersenne primes such that
  $|M_n| > \max(N(\delta), 2^{\frac{2}{e^\gamma
      \epsilon}})$ %
  and $|M_{n+1}| > |M_n|^{1+\epsilon}$. We have
  $\mu(|M_n|) \leq (1+\delta) (e^\gamma \log |M_n|)$, and since there
  are no Mersenne primes between $M_n$ and $M_{n+1}$, %
  $\mu(|M_{n+1}|) \leq 1 + (1+\delta) (e^\gamma \log
  |M_n|)$.\footnote{This follows since there can be at most one
    Mersenne prime of each length, and so in particular $M_{n+1}$ is
    the sole Mersenne prime of length $|M_{n+1}|$.} We also have that
  \begin{align*}
        \mu(|M_{n+1}|) &\geq (1-\delta) (e^\gamma \log |M_{n+1}|) \\
        &\geq (1-\delta) (e^\gamma \log (|M_n|^{1+\epsilon})) \\
        & = (1-\delta) (1+\epsilon) (e^\gamma \log |M_n|).
  \end{align*}
  Now note that
  \begin{align*}
\MoveEqLeft[9] 1 + (1+\delta) (e^\gamma \log |M_n|) - (1-\delta) (1+\epsilon)
  (e^\gamma \log |M_n|) \\[4pt]
    & =     1 + e^\gamma (\log |M_n|) ( (1+\delta) - (1-\delta)(1+\epsilon)) \\
       & < 1 + e^\gamma (\log |M_n|) ((1+\frac{\epsilon}{2(\epsilon + 2)}) -
            (1-\frac{\epsilon}{2(\epsilon + 2)})(1+\epsilon)) \\
      &  = 1 - e^\gamma (\log |M_n|)(\frac{\epsilon}{2}).
  \end{align*}
  For
  $|M_n| > 2^{\frac{2}{e^\gamma \epsilon}}$ we have
  $1 - e^\gamma (\log |M_n|)(\frac{\epsilon}{2}) < 0$, and thus for
  $|M_n| > 2^{\frac{2}{e^\gamma \epsilon}}$ we also have
  $1 + (1+\delta) (e^\gamma \log |M_n|) < (1-\delta) (1+\epsilon)
  (e^\gamma \log |M_n|).$
  This last inequality yields a contradiction as we have also shown
  $(1-\delta) (1+\epsilon)
  (e^\gamma \log |M_n|) \leq \mu(|M_{n+1}|) \leq 1 + (1+\delta) (e^\gamma \log
  |M_n|)$.

  So for any $\epsilon > 0$ there are only finitely many $n$ such that
  the consecutive Mersenne primes $M_n$ and $M_{n+1}$ have
  $|M_{n+1}| > |M_n|^{1+\epsilon}$. Let $n_0$ be the least integer
  such that for all $n>n_0$, $|M_{n+1}| \leq |M_n|^{1+\epsilon}$. The
  set of Mersenne primes
  $\{M_i \condition i>n_0\}$ is an $n^{1 + \epsilon}$-nongappy,
  P-printable subset of the primes.
\end{proof}

\begin{corollary}\label{c:lpw-containment}
  If the Lenstra--Pomerance--Wagstaff Conjecture holds, then
  $\upleq{\bigo(1) + \log \log n} \subseteq \rcnb{\primes}$ (indeed,
  $\upleq{\bigo(1) + \log \log n} \subseteq \rcnb{\mersenneprimes}$).
\end{corollary}
\begin{proof}
    Assume that the Lenstra--Pomerance--Wagstaff Conjecture holds.
    Since $2^{1/2} > 1$, by Theorem~\ref{t:lpw} the Mersenne primes have an $n^{2^{1/2}}$-nongappy, P-printable subset.
    The conditions of Theorem~\ref{t:n^k} are satisfied with $k = 2^{1/2}$, and so we have $\upleq{\bigo(1) + \frac{\log\log n}{2\log(2^{1/2})}} \subseteq \rcnb{\mersenneprimes}$, and hence $\upleq{\bigo(1) + \log\log n} \subseteq \rcnb{\mersenneprimes} \subseteq \rcnb{\primes}$.
\end{proof}

We will soon turn to discussing more notions of nongappiness and what
containment theorems hold regarding them.  However, to support one of
those notions, we first define a function that will arise naturally in
Theorem~\ref{t:tradeoff}.

\begin{definition}\label{d:logcircledast}
  For any $\alpha \in \mathbb R$, $\alpha > 0$,
  $\log^\circledast(\alpha)$ is the largest natural number $k$ such
  that $\log^{[k]}(\alpha) \geq k$. We define $\log^\circledast(0)$ to
  be $0$.
\end{definition}

For $\alpha > 1$, taking $k = 0$ satisfies
$\log^{[k]}(\alpha) \geq k$. Also, for all $\ell > \log^*(\alpha)$,
$\log^{[\ell]}(\alpha) < \log^{[\log^*(\alpha)]}(\alpha) \leq 1 \leq \ell$,
and so no $\ell > \log^*(\alpha)$ can be used as the $k$ in the definition above.
So there is at least
one, but only finitely many $k$ such that $\log^{[k]}(\alpha) \geq k$,
which means that
$\log^\circledast(\alpha)$ is well-defined.  Notice that using the
definition of $\log^\circledast(\alpha)$ and the above, we get
$\log^\circledast(\alpha) \leq \log^*(\alpha)$ when $\alpha > 1$.  For
$\alpha \leq 1$, 0 is the only natural number for which the condition
from the definition holds, and so $\log^\circledast(\alpha) = 0$ if
$\alpha \leq 1$.  Thus for $\alpha \leq 1$,
$\log^\circledast(\alpha) = \log^*(\alpha)$.  (Since
Definition~\ref{d:logcircledast}'s first sentence allows values on the
open interval $(0,1)$, one might worry that the fact that we have
globally redefined $\log(\cdot)$ to implicitly be
$\log(\max(1,\cdot))$ might be changing what
$\log^\circledast(\alpha)$ evaluates to.  However, it is easy to see
that, with or without the max, what this evaluates to in the range
$(0,1)$ is 0, and so our implicit max is not changing the value of
$\log^\circledast$.)

We are using a ``variant star'' notation for $\log^\circledast$
because it in fact is related both definitionally and in terms of
value to $\log^*$.  As to its definition, $\log^\circledast$ can
alternatively be defined as the following, which in form looks far
closer to the definition of $\log^*$ than does the version in
Definition~\ref{d:logcircledast} above: ``For any
$\alpha \in \mathbb R$, $\alpha > 0$, $\log^\circledast(\alpha)$ is
${{-}1}$ plus the smallest natural number $k$ such that
$\log^{[k]}(\alpha) < k$.
We define $\log^\circledast(0)$ to be $0$.''  As to
the relationship of its values to those of $\log^*$, we have the
following theorem.

\begin{theorem} \label{t:ilog} For all $\alpha \geq 0$,
  $\log^*(\alpha) - \log^*(\log^*(\alpha) + 1) - 1 \leq
  \log^\circledast(\alpha) \leq \log^*(\alpha)$.
\end{theorem}

\begin{proof}
  We have for all $\alpha \geq 0$,
  $\log^\circledast(\alpha) \leq \log^*(\alpha)$, which follows from
  the discussion before this theorem.  Take any $\alpha \geq 0$.
  From the definition of $\log^\circledast$ it follows that
  $\log^{[\log^\circledast(\alpha)]}(\alpha) \leq
  2^{\log^\circledast(\alpha) + 1}$, for if not,
  then we must have
  $\log^{[\log^\circledast(\alpha)]}(\alpha) >
  2^{\log^\circledast(\alpha) + 1}$, which
  (taking the logarithm of both sides)
  implies that
  $\log^{[\log^\circledast(\alpha) + 1]}(\alpha) >
  \log^\circledast(\alpha) + 1$, contradicting the fact that
  $\log^\circledast(\alpha)$ is the greatest number for which such an
  inequality holds.
  Notice that for any $x$, since $\log^*(x)$ is the smallest number
  of logarithms one needs to apply to $x$ to obtain
  a result less than or equal to 1, we have that for any $k \leq \log^*(x)$,
  $\log^*(x) = k + \log^*(\log^{[k]}(x))$.
  Plugging in $x = \alpha$ and $k = \log^\circledast(\alpha)$ we get
  $\log^*(\alpha) = \log^\circledast(\alpha) +
  \log^*(\log^{[\log^\circledast(\alpha)]}(\alpha)) \leq
  \log^\circledast(\alpha) + \log^*(2^{\log^\circledast(\alpha) + 1})
  \leq \log^\circledast(\alpha) + \log^*(2^{\log^*(\alpha) + 1}) =
  \log^\circledast(\alpha) + \log^*(\log^*(\alpha) + 1) + 1$,
  where the second inequality holds from the upper bound.
  Rearranging gives us our lower bound.
\end{proof}

Theorem~\ref{t:ilog}'s upper bound leaves open the
possibility that
$\log^\circledast(\alpha)$ and  $\log^*(\alpha)$
might  be the same, or if not then at least that the former might be less than
the latter by no more than some global constant.  However, we
now will prove that this is not the case.
That is, we will show as Theorem~\ref{t: log i.o.} that
there is an infinite collection $\mathcal{T}$
of natural numbers such that for no constant $d'$ does it hold,
on every element of the collection, that $\log^\circledast$ is at most $d'$
less than $\log^*$.  In fact, we will show a slightly stronger result
than that.

First, we introduce some useful mathematical notions.

\begin{definition}[see~\cite{goo:j:transfinite,ruc:b:infinity-mind,ney:unpub:tetration}]
    For each $n \in \naturalnumber$, the $n$th tetration of 2 is defined inductively by
    \begin{equation*}
        {^n 2} = \begin{cases}
            1 & n = 0 \\
            2^{(^{(n-1)} 2)} & n > 0.
        \end{cases}
    \end{equation*}
\end{definition}
Here we are using the so-called ``Rudy Rucker notation'' for tetration
introduced by Goodstein~\cite{goo:j:transfinite} and popularized by Rucker~\cite{ruc:b:infinity-mind}.

It is easy to see that the $n$th tetration of 2, $n\in\naturalnumber$, is exactly
$2^{2^{\iddots^2}}$ where there are $n$ 2s in the tower (and, as a convention, we
view a height
zero tower of 2s as evaluating to the value 1).
Since tetration is injective, it has an inverse defined on towers of 2s.

\begin{definition}[see~\cite{rub-rom:unpub:ackerman-slog}] \label{d:slog}
    Let $\mathcal{T}$ be the set $\{{^n 2} \condition n \in \naturalnumber\}$. The (base 2) superlogarithm, $\slog : \mathcal T \to \naturalnumber$ is the inverse operation to tetration. That is, for any $N = {^n 2}$, $\slog N = n$.
\end{definition}

It is easy to see that $\slog$ is increasing. While Definition~\ref{d:slog}
only defines $\slog$ for towers of 2s, we can extend it to a function from $\realnumberatleastone$ to the nonnegative real numbers as follows. First, we extend tetration of 2 to a function $^t 2$ from the nonnegative real numbers to $\realnumberatleastone$ via linear interpolation.%
\footnote{For any $f : \naturalnumber \to \realnumberatleastone$, the linear interpolation of $f$ is the function $\tilde f : \realnumber^{\geq 0} \to \realnumberatleastone$ given by $\tilde f(x) = (1- (x - \floor{x}))f(\floor{x}) + (x - \floor{x}) f(\floor{x} + 1)$.}
Note that this extension is surjective and increasing, so it has an inverse $\widetilde{\slog} : \realnumberatleastone \to \realnumber^{\geq 0}$. This inverse agrees with $\slog$ on towers of 2s, so we may safely write $\slog$ in place of $\widetilde{\slog}$.

With these notions in hand, we prove the following ``infinitely often'' superconstant separation result between $\log^\circledast$ and
$\log^*$.
\begin{theorem} \label{t: log i.o.}
  For $n \in \naturalfromtwo$, $\log^*({^n 2}) - \log^\circledast({^n 2}) \geq \slog(\frac{2}{3} n)$.
\end{theorem}

\begin{proof}
Notice that the function $t + \slog t$ from $\realnumberatleastone$ to $\realnumberatleastone$ is
increasing and surjective, and thus has an inverse.%
\footnote{That $t + \slog t$ from $\realnumberatleastone$ to $\realnumberatleastone$
is surjective follows from the basic facts from mathematical analysis
that increasing, surjective (real) functions are
continuous, and that the range of an increasing, continuous function is an interval.
The first fact implies that
our extension of $\slog$ is continuous, and hence that $t + \slog t$ is increasing and continuous. Since $t + \slog t$
on our domain attains a minimum value of 1 and is unbounded, its range is $[1,\infty)$, which is exactly what
it means for it to be surjective onto $\realnumberatleastone$.}
Let $\mathfrak s : \realnumberatleastone \to \realnumberatleastone$ be this inverse.  We will
use the following lemma.

\begin{lemma}\label{l:slog}
    For all $n \in \naturalnumberpositive$, $\log^*({^{n}2}) - \log^\circledast({^{n}2}) = (\mathfrak s(n) - \floor{\mathfrak s(n)}) + \slog(\mathfrak s(n))$.
\end{lemma}

\begin{proof}[Proof of Lemma~\ref{l:slog}]
    From the definition of $\log^\circledast$, we have
        $\log^\circledast({^n 2}) = \max\{k \in \naturalnumber \condition \log^{[k]}({^n 2}) \geq k\}$.
        Since $n \geq 1$, ${^n 2} \geq 2$, which means $\log^{[1]}({^n 2}) \geq 1$. This means that the
        max in the previous equation is at least 1, and so we can let the max run over $\naturalnumberpositive$
        without changing the value. Also, notice that $\log^{[n]}({^n 2}) = 1$, which, since $n \geq 1$, means
        that for all $\ell > n$, $\log^{[\ell]}({^n 2}) < \ell$. So any $k$ in the set we are maxing over
        must be at most $n$, and thus $\log^{[k]}({^n 2}) = {^{n-k} 2}$. Hence $\log^\circledast({^n 2})
        = \max\{k \in \naturalnumberpositive \condition {^{n-k} 2} \geq k\}
        = \max\{k \in \naturalnumberpositive \condition n - k \geq \slog k\}
        = \max\{k \in \naturalnumberpositive \condition k + \slog k \leq n\}$.

    Since $\mathfrak s$ is increasing with inverse $k + \slog k$, we get
    $\{k \in \naturalnumberpositive \condition k + \slog k \leq n \} = \{k \in \naturalnumberpositive \condition k \leq \mathfrak s(n)\}$,
    and thus $\log^\circledast(^n 2) = \floor{\mathfrak s(n)}$.

    On the other hand, we have $\log^*(^n 2) = n = \mathfrak s(n) + \slog (\mathfrak s(n))$, and thus $\log^*(^n 2) - \log^\circledast(^n 2) = (\mathfrak s(n) - \floor{\mathfrak s(n)}) + \slog (\mathfrak s(n))$. This concludes the proof
    of this lemma.
\end{proof}

Let us
get a handle on the function $\mathfrak s$.
Notice that for real numbers $t \geq 0$ we have $^t2 \geq 2t$, since the inequality holds when $t$ is a natural number, and taking linear interpolations of both sides preserves the inequality.
Changing variables, we get that for all $t \geq 0$, $^{t/2} 2 \geq t$.
Applying $\slog$ when defined we get that for all $t \geq 1$, $t/2 \geq \slog t$.
From the definition of $\mathfrak s$ as the inverse of $t + \slog t$, for each $n \in \naturalnumberpositive$ we have
$n = \mathfrak s(n) + \slog (\mathfrak s(n))$, which, by the inequality
just mentioned
is less than or equal to $\frac{3}{2} \mathfrak s(n)$,
so $\mathfrak s(n) \geq \frac{2}{3} n$.
Combining this with Lemma~\ref{l:slog} and using the fact that $\slog$ is increasing on its (now extended) domain $\realnumberatleastone$,
we get that for all $n \in \naturalfromtwo$ we have
$\log^*({^{n}2}) - \log^\circledast({^{n}2}) = (\mathfrak s(n) - \floor{\mathfrak s(n)}) + \slog(\mathfrak s(n))
\geq \slog(\mathfrak s(n))
\geq \slog(\frac{2}{3} n)$, thus establishing the theorem's claim.
The only reason the previous sentence, and the theorem's
statement,
exclude $n=1$ and start at $n=2$ is that
$\slog$ is defined only on reals greater than or equal to 1, and thus simply is not defined at $\frac{2}{3}n$ when $n = 1$.
\end{proof}

We now return to our study of nongappy sets, where the notion of
$\log^{\circledast}$ will play an important role.

\begin{definition}\label{def:polylog-nongappy}
  A nonempty set $S \subseteq \naturalnumberpositive$ is $\bigo(n \log n)$-nongappy if
    $(\exists f \in \bigo(n \log n))(\forall m \in S)(\exists m' \in
    S) [m'>m \land |m'| \leq f(|m|)]$.
\end{definition}

Definitions of $n^{(\log n)^k}$-nongappy for any constant $k \in \naturalnumberpositive$ and $2^n$-nongappy are provided
via Definition~\ref{d:nongappy-main}, since $n^{(\log n)^k}$ and $2^n$ are
each a single function, not a collection of
functions.\footnote{Note that $n^{(\log n)^k}$-nongappiness does not involve
evaluating $0^0$ even though it might at first seem to because Definition~\ref{d:nongappy-main},
which is used to define the notion, restricts the domain of ``$F$'' to $\realnumberpositive$, and because
$k$ is a positive natural number.}
Those two notions, along with the notion defined in
Definition~\ref{def:polylog-nongappy}, will be the focus of
Theorem~\ref{t:tradeoff}.  That theorem obtains the containments
related to those three notions of nongappiness.
As one would expect,
as the allowed gaps become larger
the corresponding $\up$ classes become
more restrictive in their ambiguity bounds.
\begin{theorem} \label{t:tradeoff} Let $T$ be a subset of
  $\naturalnumberpositive$.
  \begin{enumerate}
  \item\label{p:tradeoff0} If $T$ has
    an $\bigo(n\log n)$-nongappy, P-printable subset, then
    $\upleq{\bigo(\sqrt{\log n})} \subseteq \rcnb{T}$.
  \item\label{p:tradeoff2-new} For all $k \in \naturalnumberpositive$, if $T$ has an
    $n^{(\log n)^k}$-nongappy, P-printable subset, then
    $\upleq{\bigo(1) + \frac{1}{\ceiling{\log(k+1) + 1}}\log \log \log n} \subseteq
    \rcnb{T}$.\footnote{Earlier versions of this paper claimed that if $T$ has an $n^{(\log n)^{\bigo(1)}}$-nongappy (which is defined analogously to other notions of nongappiness involving big-Os, e.g., Def.~\ref{def: n+O(1)-nongappy}), P-printable subset then $\upleq{\bigo(1) + \frac{1}{3}\log^{[4]}(n)} \subseteq \rcnb{T}$ \cite[Theorem~4.19, Part~3]{hem-juv-nad-phi:c:ics} \cite[Theorem~4.23, Part~3]{hem-juv-nad-phi:t4:itc}, although those versions either pointed to or included a flawed proof.
    An anonymous ACT~TOCT referee both spotted the flaw and generously suggested a tighter inequality that, when used in the proof, would improve the result.
    By further tightening that inequality into an identity, we were able to prove the stronger result that appears here, namely, part~\ref{p:tradeoff2-new} of Theorem~\ref{t:tradeoff}.
    The current result implies the old statement because if a set $T$ has an $n^{(\log n)^{\bigo(1)}}$-nongappy, P-printable subset then there exists $k \in \naturalnumberpositive$ such that $T$ has an $n^{(\log n)^k}$-nongappy, P-printable subset, and because $\upleq{\bigo(1) + \frac{1}{3}\log^{[4]}(n)} \subseteq \upleq{\bigo(1) + C \log^{[3]}}(n)$ for any constant $C > 0$. The latter holds because for any $C$ there exists $N$ such that $C\log^{[3]}(n) \geq \frac{1}{3}\log^{[4]}(n)$ for all $n > N$, and so if a machine $M$ witnesses $L \in \upleq{\bigo(1) + \frac{1}{3}\log^{[4]}(n)}$ then the machine $M'$ that, on inputs of length at most $N$, memorizes whether to accept or reject, and, on inputs of length greater than $N$, simulates $M$, witnesses $L \in \upleq{\bigo(1) + C\log^{[3]}(n)}$.}
  \item\label{p:tradeoff3} If $T$ has a $2^n$-nongappy, P-printable
    subset $S$, then
    $\upleq{\max(1, \floor{\frac{ \log^{\circledast} n}{\lambda}})} \subseteq
    \rcnb{T}$ (and so certainly also
    $\upleq{\max(1, \floor{\frac{ \log^*(n) - \log^*(\log^*(n) + 1) - 1)}
      {\lambda}})} \subseteq \rcnb{T}$), where $\lambda = 4 + \min_{s \in S, |s| \geq 2}(|s|)$.
  \end{enumerate}
\end{theorem}

\begin{proof}
We prove each of the three parts of the theorem separately.
{\bf (Part~\ref{p:tradeoff0})}\quad
  Suppose $T \subseteq \naturalnumberpositive$ has an $\bigo(n\log n)$-nongappy, P-printable subset.
  It follows from the definition of $\bigo(n\log n)$-nongappy that
  there is some
  $k \in \naturalnumberpositive$ such that $T$ has a
  $kn \log n$-nongappy, P-printable subset. Set $F: \realnumberpositive \to \realnumberpositive$ to be $F(t) = kt \log t$.
  The conditions from
  Theorem~\ref{t:meta2} are satisfied by $F(t)$ as for all
  $t \geq 4$, $F(t) = kt \log t \geq t+2$ and
  $(\forall c \in \realnumberatleastone)$
  [$cF(t) = c kt \log t \leq c kt \log ct = F(ct)]$, and $F$ is
  nondecreasing on $\{t \in \realnumberpositive \condition t \geq 4\}$.
  Let $\lambda = 4+|s|$ where $s$ is the smallest element of the
  $kn\log n$-nongappy, P-printable subset of $T$ such that the conditions on $F$
  hold for all $t \geq |s|$, i.e.,
  $s$ is the smallest element of the
  $kn\log n$-nongappy, P-printable subset of $T$ such that $|s| \geq 4$.
  For every function $g: \naturalnumber \to \realnumberatleastone$ satisfying $g(n) = \bigo(\sqrt{\log n})$
  it is easy to see that there exists a number $d$ such that $(\forall n \in \naturalnumber) [g(n) \leq d(\sqrt{\log n}+1)]$ . Thus
  $\upleq{g(n)} \subseteq \upleq{d(\sqrt{\log n}+1)} = \upleq{\floor{d(\sqrt{\log n}+1)}}$. The function
  $j(n) = \floor{d(\sqrt{\log n}+1)}$ satisfies the conditions of
  Theorem~\ref{t:meta2} as $j(n)$ can be computed in time polynomial
  in the value $n$
  (since $\floor{d(\sqrt{\log n}+1)}$ can be computed by doing a linear search for the largest natural number $i$ such that $2^{i^2} \leq n^{d^2}$),
  and $j(n)$ has value at most polynomial in the
  value $n$.  Applying Theorem~\ref{t:meta2}, to prove that
  $\upleq{j(n)} \subseteq \rcnb{T}$ it suffices to show that there is
  some $\beta$ such that ${F^{[j(n)]}(j(n)\lambda)} = \bigo(n^\beta)$.

To this end, we show, via induction on $\ell$, that for all $\ell \in \naturalnumberpositive$ and real $t \geq 1$,
\begin{equation}\label{eq: 4.21_pt1_inequality}
    F^{[\ell]}(t)\leq \ell! k^\ell t [\log((\ell - 1)! k^{\ell - 1} t)]^\ell.
\end{equation}
The base case $\ell = 1$ is an equality since the right-hand side for $\ell = 1$ is exactly the definition of $F(t)$.
Assume that Eq.~\ref{eq: 4.21_pt1_inequality} holds for some $\ell \geq 1$. Then
\begin{align*}
    F^{[\ell + 1]}(t) &= k F^{[\ell]}(t) \log(F^{[\ell]}(t)) \\
	&\leq
	k(\ell! k^\ell t [\log((\ell -1)! k^{\ell - 1} t)]^\ell) \log(\ell! k^\ell t [\log((\ell -1)! k^{\ell - 1} t)]^\ell) \\
	&= \ell! k^{\ell +1}t [\log((\ell -1)! k^{\ell - 1} t)]^\ell (\log(\ell! k^\ell t) + \ell \log\log((\ell -1)! k^{\ell - 1} t)) \\
    &\leq \ell! k^{\ell+1} t [\log((\ell - 1)!k^{\ell-1}t)]^ \ell \cdot \log(\ell! k^\ell t) \\
	&\leq (\ell + 1)! k^{\ell + 1} t \log(\ell! k^\ell t)^{\ell + 1},
\end{align*}
closing the induction.

Applying Eq.~\ref{eq: 4.21_pt1_inequality} with $\ell = j(n)$ and $t = j(n)\lambda$ and using $j(n)! \leq j(n)^{j(n)}$ we get
\begin{equation} \label{eq:pt1_inequality_applied}
    F^{[j(n)]}(j(n) \lambda) \leq
    j(n)^{j(n)} k^{j(n)} j(n) \lambda [\log(j(n)^{j(n)} k^{j(n)} j(n) \lambda)]^{j(n)}.
\end{equation}
For sufficiently large $n$ we have $j(n) \leq C\sqrt{\log n}$ for some constant $C \in \naturalnumberpositive$ that depends on $d$. So we have
\begin{align*}
    j(n)^{j(n)} &\leq C^{C\sqrt{\log n}} \cdot (\log n)^{C\sqrt{\log n}} \\
    &\leq 2^{C\log C \cdot \sqrt{\log n}} \cdot 2^{\log\log n \cdot C \sqrt{\log n}}.
\end{align*}
For sufficiently large $n$, $\log\log n \leq \sqrt{\log n}$, so for large $n$ the second quantity in the multiplication is at most $n^C$.
The first quantity is at most $n^{C \log C}$ since $\sqrt{\log n} \leq \log n$.
Letting $C' = C\log C + C$ for convenience, for sufficiently large $n$, $j(n)^{j(n)} \leq n^{C'}$.
Since $k$ is a constant while $j$ tends to infinity with $n$, for large $n$, $k^{j(n)} \leq j(n)^{j(n)} \leq n^{C'}$.
Finally, since $C' \geq C \geq 1$, for all $n$, $\sqrt{\log n} \leq n^{C'}$.
Plugging everything into Eq.~\ref{eq:pt1_inequality_applied} we get that for all sufficiently large $n$,
\begin{align*}
    F^{[j(n)]}(j(n) \lambda) &\leq n^{2C'} \cdot C\lambda \sqrt{\log n} \cdot [\log(n^{2C'} \cdot C\lambda \sqrt{\log n})]^{C\sqrt{\log n}} \\
    &\leq C\lambda n^{3C'} \cdot [\log(C\lambda n^{3C'})]^{C\sqrt{\log n}} \\
    &= C\lambda n^{3C'} \cdot 2^{\log\log(C\lambda n^{3C'}) \cdot C \sqrt{\log n}}.
\end{align*}
Notice that for large $n$, $\log\log(C\lambda n^{3C'}) = \log(\log(C\lambda) + 3C'\log n) \leq \log(4C' \log n) = \log(4C') + \log\log n \leq 2\log\log n \leq \sqrt{\log n}$.
Thus we have that for large $n$,
\begin{align*}
    F^{[j(n)]}(j(n) \lambda) &\leq C\lambda n^{3C'} \cdot 2^{\sqrt{\log n} \cdot C \sqrt{\log n}} \\
    &= C\lambda n^{3C' + C}.
\end{align*}
Hence there exists a $\beta$ (namely, this constant $3C' + C$) such that $F^{[j(n)]}(j(n) \lambda) = \bigo(n^\beta)$.
Thus for any function $g: \naturalnumber \to \realnumberatleastone$ satisfying
  $g(n)=\bigo(\sqrt{\log n})$, there exists a function
  $j$ such that
  $\upleq{g(n)} \subseteq \upleq{j(n)} \subseteq \rcnb{T}$.

{\bf (Part~\ref{p:tradeoff2-new})}\quad
Fix $k \in \naturalnumberpositive$.
Suppose $T \subseteq \naturalnumberpositive$ has an $n^{(\log n)^k}$-nongappy, P-printable subset.
Set $F: \realnumberpositive \to \realnumberpositive$ to be $F(t) = t^{(\log t)^k}$.
The conditions from Theorem~\ref{t:meta2} are satisfied by $F$ as for all $t \geq 4$, $F(t) \geq t+2$ and
$(\forall c \in \realnumberatleastone)[cF(t) = c t^{(\log t)^k} \leq (ct)^{(\log (ct))^k} = F(ct)]$,
and $F$ is nondecreasing on $\{t \in \realnumberpositive \condition t \geq 4\}$.
Let $\lambda = 4+|s|$ where $s$ is the smallest element of the $n^{(\log n)^k}$-nongappy, P-printable subset of $T$ such that the conditions on $F$ hold for all $t \geq |s|$, i.e.,
$s$ is the smallest element of the $n^{(\log n)^k}$-nongappy, P-printable such that $|s| \geq 4$.
For every function $g: \naturalnumber \to \realnumberatleastone$ satisfying
$g(n) = \bigo(1) + \frac{1}{\ceiling{\log(k+1)+1}} \log \log \log n$ there exists $d \in \naturalnumberpositive$ such that
$g(n) \leq d + \frac{1}{\ceiling{\log(k+1)+1}} \log\log\log n$ and hence $\upleq{g(n)} \subseteq \upleq{d + \frac{1}{\ceiling{\log(k+1)+1}} \log\log\log n} = \upleq{\floor{d + \frac{1}{\ceiling{\log(k+1)+1}} \log\log\log n}}$.
The function $j(n) = \floor{d + \frac{1}{\ceiling{\log(k+1)+1}} \log \log \log n}$ can be computed in time polynomial in the value $n$
since $\floor{\frac{\log\log\log n}{\ceiling{\log(k+1)+1}}}$ can be computed by doing a linear search for the largest natural number $i$ such that $2^{2^{2^{\ceiling{\log(k+1)+1}i}}} \leq n$ (the ceiling can be hardcoded since $k$ is a constant).
Also, $j(n)$ has value at most polynomial in the value $n$.
So $j(n)$ satisfies the conditions of Theorem~\ref{t:meta2}.
Applying Theorem~\ref{t:meta2}, to prove that $\upleq{j(n)} \subseteq \rcnb{T}$ it suffices to show that for some $\beta \in \naturalnumber$, $F^{[j(n)]}(j(n) \lambda) = \bigo (n^{\beta})$.

We first show that for all $\ell \in \naturalnumberpositive$ and $t \in \realnumberpositive$, $F^{[\ell]}(t) = t^{(\log t)^{(k+1)^\ell-1}}$.
We do so by induction on $\ell$.
Notice that for all $t \in \realnumberpositive$, $F^{[1]}(t) = t^{(\log t)^k} = t^{(\log t)^{(k+1)^1-1}}$, so the claim holds for $\ell = 1$.
Assume for induction that the identity holds for some $\ell \geq 1$.
Fix some $t \in \realnumberpositive$ and let $t' = t^{(\log t)^{(k+1)^\ell-1}}$.
Using the inductive hypothesis, $F^{[\ell+1]}(t) = F(F^{[\ell]}(t)) = F(t') = t'^{(\log t')^k}$.
We have
$(\log t')^k = ((\log t)^{(k+1)^\ell-1} \log t)^k = (\log t)^{k(k+1)^\ell}$,
and thus
\begin{equation} \label{eq:iter-ineq-pre-binomial}
    F^{[\ell+1]}(t) = (t^{(\log t)^{(k+1)^\ell-1}})^{(\log t)^{k(k+1)^\ell}} = t^{(\log t)^{k(k+1)^\ell + (k+1)^\ell - 1}}.
\end{equation}
Using the binomial theorem,
\begin{align*}
    k(k+1)^\ell + (k+1)^\ell - 1 &= -1 + \sum_{0 \leq i \leq \ell} {\ell \choose i} k^{i+1} + \sum_{0 \leq i \leq \ell} {\ell \choose i} k^i \\
    &= -1 + \sum_{1 \leq i \leq \ell+1} {\ell \choose i-1} k^i + \sum_{0 \leq i \leq \ell} {\ell \choose i} k^i \tag{by reindexing} \\
    &= -1 + k^{\ell+1} + 1 + \sum_{1 \leq i \leq \ell} \left[{\ell \choose i-1} + {\ell \choose i}\right] k^i \\
    &= k^{\ell + 1} + \sum_{1 \leq i \leq \ell} {\ell + 1 \choose i} k^i \\
    &= (k+1)^{\ell+1} - 1,
\end{align*}
which, when substituted back into Eq.~\ref{eq:iter-ineq-pre-binomial}, gives us $F^{[\ell+1]}(t) = t^{(\log t)^{(k+1)^{\ell+1}-1}}$, completing the induction.

We now use this identity to prove that there exists $\beta$ such that $F^{[j(n)]}(j(n)\lambda) = \bigo(n^\beta)$.
For convenience, let $t_n = j(n)\lambda$.
Notice that since for all $n$, $j(n)\lambda > 0$ and $j(n) \in \naturalnumberpositive$, we can apply the identity we just proved to get $F^{[j(n)]}(t_n) = t_n^{(\log t_n)^{(k+1)^{j(n)}-1}}$.
To complete the proof, it suffices to show that there is a constant $\beta$ such that for sufficiently large $n$, the expression on the right-hand side is at most $n^\beta$.
Taking the log of both sides twice, it suffices to show that there exists a constant $\beta$ such that for large enough $n$, $(k+1)^{j(n)}\log\log(t_n) \leq \log\log n + \log \beta$.
Plugging in the definitions of $j$ and $t_n$,
\begin{align*}
    (k+1)^{j(n)}\log\log(t_n) &\leq (k+1)^{d + \frac{1}{\ceiling{\log(k+1)+1}}\log^{[3]}(n)} \cdot \log\log(d\lambda + \frac{d}{\ceiling{\log(k+1)+1}}\log^{[3]}(n)) \\
    &= (k+1)^d \cdot 2^{\frac{\log(k+1)}{\ceiling{\log(k+1)+1}}\log^{[3]}(n)} \cdot \log\log(d\lambda + \frac{d}{\ceiling{\log(k+1)+1}}\log^{[3]}(n)).
\end{align*}
It is easy to see that for large enough $n$, $\log\log(d\lambda + \frac{d}{\ceiling{\log(k+2)+1}}\log^{[3]}(n)) \leq 2\log^{[5]}(n)$
(asymptotically, the leading-order term is $\log^{[5]}(n)$ on the left and $2\log^{[5]}(n)$ on the right).
For convenience, let $\epsilon = 1 - \frac{\log(k+1)}{\ceiling{\log(k+1)+1}}$.
We have that for sufficiently large $n$,
\begin{align*}
    (k+1)^{j(n)}\log\log(t_n) &\leq (k+1)^d \cdot (\log\log n)^{\frac{\log(k+1)}{\ceiling{\log(k+1)+1}}} \cdot 2\log^{[5]}(n),
    \\
    &=  2(k+1)^d \cdot (\log\log n)^{1-\epsilon} \cdot \log^{[5]}(n). %
\end{align*}
For large $n$, $\log^{[5]}(n) \leq (\log\log n)^{\epsilon/2}$, and so the expression above is bounded above by $C (\log\log n)^{1-\epsilon/2}$ where $C$ is a constant that depends on $k$ and $d$.
For all $\beta \geq 1$ and for sufficiently large $n$, this quantity is at most $\log\log n + \log \beta$.
Thus there exists a $\beta$ (namely, any $\beta \geq 1$) such that $F^{[j(n)]}(j(n)\lambda) = \bigo(n^\beta)$.

Putting everything together, we have showed that for every $g : \naturalnumber \to \realnumberatleastone$ satisfying $g(n) = \bigo(1) + \frac{1}{\ceiling{\log(k+1)+1}} \log\log\log n$ there exists a function $j$ such that $\upleq{g(n)} \subseteq \upleq{j(n)} \subseteq \rcnb{T}$.

{\bf (Part~\ref{p:tradeoff3})}\quad
  Suppose $T \subseteq \naturalnumberpositive$ has a $2^n$-nongappy, P-printable subset $S$.
  Set $F: \realnumberpositive \to \realnumberpositive$ to be $F(t) = 2^t$.
  The conditions from
  Theorem~\ref{t:meta2} are satisfied by $F(t)$ as for all
  $t \geq 2$, $F(t) \geq t+2$ and
  $(\forall c \in \realnumberatleastone)$
  [$cF(t) = c \cdot 2^t \leq 2^{ct} = F(ct)]$, and $F$ is
  nondecreasing on $\{t \in \realnumberpositive \condition t \geq 2\}$.
  Let $\lambda$ be as defined in the theorem statement, i.e., $\lambda = 4 + \min_{s \in S, |s| \geq 2}(|s|)$.
  Notice that this is equal to $4+|s|$ where $s$ is the smallest element of $S$ where the conditions on $F$ hold, and so $\lambda$ is as in Theorem~\ref{t:meta2}.
  Let $j: \naturalnumber \to \naturalnumberpositive$ be
  $j(n) = \max(1, \floor{\lambda^{-1} \log^{\circledast} (n)})$.
  Since $\log^{\circledast}$ can be computed in polynomial time, the function $j(n)$
  can be computed in time at most polynomial in the value $n$ and also will have
  value at most polynomial in the value $n$.
  Applying Theorem~\ref{t:meta2}, to show that $\upleq{\max(1, \floor{\frac{\log^{\circledast}n}{\lambda}})} \subseteq \rcnb{T}$ it is enough to show that
  $F^{[j(n)]}(j(n)\lambda) = \bigo(n)$.
  It suffices to show that for all sufficiently large $n$, $F^{[j(n)]}(j(n)\lambda) \leq n$.
  Since $\log^\circledast(n) \to \infty$ as $n \to \infty$, for large enough $n$ we have $\lambda^{-1} \log^\circledast(n) \geq 1$ and hence $j(n) = \floor{\lambda^{-1} \log^\circledast(n)}$.
  Thus for sufficiently large $n$,
  \begin{equation*}
    F^{[j(n)]}(j(n) \lambda) \leq {\underbrace{2^{2^{\iddots^{2}}}}_{\textrm{$j(n)$}}}^{^{\log^{\circledast}(n)}}
    \leq {\underbrace{2^{2^{\iddots^{2}}}}_{\textrm{$j(n)$}}}^{^{\log^{[\log^{\circledast} (n)]}(n)}}
    \leq {\underbrace{2^{2^{\iddots^{2}}}}_{\textrm{$\log^\circledast (n)$}}}^{^{\log^{[\log^{\circledast} (n)]}(n)}}
    = n,
  \end{equation*}
  which finishes the proof.
\end{proof}

\section{Conclusions and Open Problems}
This paper applied and adapted the iterative constant-setting
technique used by Cai and Hemachandra~\cite{cai-hem:j:parity} and
Borchert, Hemaspaandra, and Rothe~\cite{bor-hem-rot:j:powers-of-two}
to a more general setting.  In particular, we generalized Borchert,
Hemaspaandra, and Rothe's notion of ``nongappiness,'' proved two
flexible metatheorems that can be used to obtain containments of
ambiguity-limited classes in restricted counting classes, and applied
those theorems to prove containments for some of the most natural
ambiguity-limited classes.  We also noted the apparent trade-off
between the nongappiness of the targets used in iterative
constant-setting, and the nondeterministic ambiguity of the classes
one can capture using those targets.
For example, beyond the containments we explicitly derived with
Theorems~\ref{t:meta1} and~\ref{t:meta2}, those two meta\-theorems
themselves seem to reflect a
trade-off between the ambiguity allowed in
an ambiguity-limited class and the smallness of gaps in a set of
natural numbers defining a restricted counting class.  One open
problem is to make explicit, in a smooth and complete fashion, this
trade-off between gaps and ambiguity.
Another open problem is to capture the relationship
between $\log^\circledast$ and $\log^*$ more tightly
than Theorems~\ref{t:ilog} and~\ref{t: log i.o.} do.

One last related open research direction is to
further study nongappy,
P-printable subsets of the primes.  We noted two sufficient conditions
for showing the existence of P-printable subsets of primes, namely
Allender's hypothesis about the probabilistic complexity class
$\mathrm{RP}$~\cite[Corollary 32 and the comment following it]{all:j:pseudorandom}
and the Lenstra--Pomerance--Wagstaff
Conjecture~\cite{pom:j:primality-testing,wag:j:mersenne-numbers}.
Furthermore, we proved that
the
Lenstra--Pomerance--Wagstaff Conjecture in fact
implies that for all
$\epsilon>0$ the primes have an $n^{1+\epsilon}$-nongappy, P-printable
subset. While finding a P-printable subset of the primes would itself
be interesting, we have shown how it would also be a useful step
towards understanding the restricted counting class defined by the
primes, namely, if one could find a suitably nongappy such set.

\paragraph{Acknowledgments}  An anonymous ACM~TOCT referee found, and in many cases suggested corrections for, flaws in a number of our proofs, and that referee also gave us proposed argument lines and suggested strengthenings that we have followed, e.g., part~\ref{p:tradeoff2-new} of Theorem~\ref{t:tradeoff} in its current form is stronger than in previous versions of the paper.
We are deeply grateful to
the anonymous
ACM~TOCT and MFCS referees,
and to 
Eric Allender, Benjamin Carleton, Kenneth Regan, and Henry Welles,
for helpful comments, corrections, information, and improvements.

\appendix
\section*{Appendix A:
Deferred Proof of Theorem~\protect\ref{t:up-to-k}}\label{a-now-nonsub-section-NOW-HACKED-BY-HAND-AS-TO-LABEL:up-to-k}

We now briefly give the simple proof of Theorem~\ref{t:up-to-k}. We
assume that the reader has already read the less simple proof of
Theorem~\ref{t:meta1}, and thus has seen that proof's  use of iterative
constant-setting.

\begin{proof}[Proof of Theorem~\ref{t:up-to-k}]
  Let $L$ be a language in $\upleq{k}$, witnessed by a machine
  $\hat{N}$.  To show $L \in \rcnb{T}$ we give a description of a NPTM
  $N$ that on every input $x$ has $\acc_N(x) \in T$ if $x \in L$ and
  $\acc_N(x) = 0$ if $x \notin L$.

  On input $x$, $N$ nondeterministically guesses an integer
  $i \in \{1,2,\ldots,k\}$, and then nondeterministically guesses
  a cardinality-$i$
  set of paths of $\hat{N}(x)$.  If all the paths guessed in a
  cardinality-$i$
  set are accepting paths, then $N$ branches into $c_i$ accepting
  paths, where the constants $c_i$ are as defined below.  Note that
  unlike the proof of Theorem~\ref{t:meta1} these constants
  $c_1, \ldots, c_k$
  do not have to be computed on the fly by $N$, but rather are
  hard-coded into $N$, so we do not need $T$ to be P-printable.

  Set $c_1$ to be the least element of $T$. Iteratively
  set, in this order, $c_2$, $c_3$,~\dots, $c_k$ as follows.  Given
  $c_1, \ldots, c_{i-1}$ set
  $b_i = \sum_{1 \leq \ell \leq i-1} c_{\ell} {i \choose \ell}$.
  Then let $a_i$ be the least element of $T$ such that $a_i \geq b_i$,
  and set $c_i = a_i - b_i$.  Our description of machine $N$ is
  complete.

  Similarly to the proof of Theorem~\ref{t:meta1}, we have set $c_i$
  to ensure that $\acc_N(x) \in T$ if $\hat{N}(x)$ accepts and
  $\acc_N(x) = 0$ if $\hat{N}(x)$ rejects.
  It is clear from the construction---keeping in mind that $\hat{N}$
  runs in nondeterministic polynomial time and the $c_i$ each will be
  fixed constants---that
  $N$ is an NPTM\@.
\end{proof}%
%
%
%

%
\newcommand{\etalchar}[1]{$^{#1}$}

\end{document}